\documentclass[11pt]{article}

\date{}
\usepackage{hyperref}
\usepackage[boxruled,lined,linesnumbered]{algorithm2e}
\usepackage{amsthm,graphicx,graphics,amsmath,amssymb}
\usepackage{comment}
\usepackage{enumitem}
\newcommand{\TT}{{\cal T}}
\newcommand{\PP}{{\cal P}}
\newcommand{\St}{{\cal S}}

\newtheorem{theorem}{Theorem}
\newtheorem{lemma}{Lemma}
\newtheorem*{remark}{Remark}

\begin{document}

\title{Near Optimal Parallel Algorithms for Dynamic DFS in\\ Undirected Graphs~\footnote{The preliminary version of this 
paper have been accepted to appear at SPAA 2017}}

\author{Shahbaz Khan 
\thanks{Dept. of CSE, IIT Kanpur, India, 
email: \texttt{shahbazk@cse.iitk.ac.in}}
\thanks{This research work was supported by
Google India under the Google India PhD Fellowship Award.}
}




\maketitle
\begin{abstract}
Depth first search (DFS) tree is a fundamental data structure for solving various 
graph problems. The classical algorithm~\cite{Tarjan72} for building a DFS tree requires 
$O(m+n)$ time for a given undirected graph $G$ having $n$ vertices and $m$ edges.
Recently, Baswana et al.~\cite{BaswanaCCK15} presented a simple algorithm for updating the 
DFS tree of an undirected graph after an edge/vertex update in $\tilde{O}(n)$
\footnote{$\tilde{O}()$ hides the poly-logarithmic factors.} time.
However, their algorithm is strictly sequential. 
We present an algorithm achieving similar bounds 
that can be  easily adopted to the parallel environment. 

In the parallel environment, a DFS tree can be computed from scratch using $O(m)$ processors in 
expected $\tilde{O}(1)$ time~\cite{AggarwalAK90} on an EREW PRAM, whereas the best 
deterministic algorithm takes $\tilde{O}(\sqrt{n})$ time~\cite{AggarwalAK90,GoldbergPV93} 
on a CRCW PRAM. Our algorithm can be used to develop optimal time (upto $poly\log n$ factors) 
deterministic parallel algorithms for maintaining fully dynamic DFS and fault tolerant DFS
of an undirected graph. 


\begin{enumerate}
	\item {\em Parallel Fully Dynamic DFS}:	\\
	Given an arbitrary online sequence of vertex or edge updates, 
	we can maintain a DFS tree of an undirected graph in $\tilde{O}(1)$ 
	time per update using $m$ processors on an EREW PRAM. 

	\item {\em Parallel Fault tolerant DFS}:\\	
	An undirected graph can be preprocessed to build a data structure of size $O(m)$, such that 
	for any set of $k$ updates (where $k$ is constant) in the graph, a DFS tree of the updated graph 
	can be computed in $\tilde{O}(1)$ 
	time using $n$ processors on an EREW PRAM.
	For constant $k$, this is also work optimal 
	(upto $poly\log n$ factors). 
\end{enumerate} 


Moreover, our fully dynamic DFS algorithm provides, in a seamless manner, 
nearly optimal (upto $poly\log n$ factors) algorithms for maintaining a DFS tree in 
the semi-streaming environment and a restricted distributed model. 
These are the first parallel, semi-streaming and distributed  algorithms for maintaining 
a DFS tree in the dynamic setting.

\end{abstract}



\section{Introduction}
Depth First Search (DFS) is a well known graph traversal technique. 
Right from the seminal work of Tarjan~\cite{Tarjan72}, DFS traversal 
has played a central role in the design of efficient algorithms for many 
fundamental graph problems, namely, 
strongly connected components, topological sorting~\cite{Tarjan76}, 
dominators in directed graph~\cite{Tarjan74},
edge and vertex connectivity~\cite{EvenT75} etc. 

Let $G=(V,E)$ be an undirected connected graph having $n$ vertices and $m$ edges.
The DFS traversal of $G$ starting from any vertex $r\in V$ 
produces a spanning tree rooted at $r$ 
called a DFS tree, in $O(m+n)$ time. 
For any rooted spanning tree of $G$, a non-tree edge of the graph
is called a {\it back edge} if one of its endpoints is an ancestor of the 
other in the tree. Otherwise, it is called a {\it cross edge}. 
A necessary and sufficient condition for any rooted spanning tree to be a 
DFS tree is that every non-tree edge is a back edge. Thus, 
many DFS trees are possible for any given graph from a given root $r$. 
However, if the traversal is performed strictly according to the order of 
edges in the adjacency lists of the graph, the resulting DFS tree will
be unique. Ordered DFS tree problem is to compute the order in which the vertices are 
visited by this unique DFS traversal. 


An algorithmic graph problem is modeled 
in a dynamic environment as follows. There is an online sequence of 
updates on the graph, and the objective is to update the solution of 
the problem efficiently after each update. 
In particular, the time taken to update the solution has to be much 
smaller than that of the best static algorithm for the problem. 
A dynamic graph algorithm is said to be {\em fully dynamic} if it handles both 
insertion and deletion updates, otherwise it is called {\em partially dynamic}. 
Another, and more restricted, variant of a dynamic environment is the 
fault tolerant environment. Here the aim is to build a compact data structure, 
for a given problem, that is resilient to failures of vertices/edges and can 
efficiently report the solution after a given set of failures.

Recently, Baswana et al.~\cite{BaswanaCCK15,BaswanaCCK16} presented a fully dynamic algorithm for maintaining 
a DFS tree of an undirected graph in $\tilde{O}(\sqrt{mn})$ time per update. 
They also presented an algorithm for updating the DFS tree after a single update in $\tilde{O}(n)$ time.
Prior to this work, only partially dynamic algorithms were known for 
maintaining a DFS tree~\cite{BaswanaC15,BaswanaK14,FranciosaGN97}. 
%



Now, major applications of dynamic graphs in the real world involve a huge amount of data, 
which makes recomputing the solution after every update infeasible. Due to this large size of data, 
it also becomes impractical for solving such problems on a single sequential machine because of 
both memory and computation costs involved. 
Thus, it becomes more significant to explore these dynamic graph problems on a 
computation model that efficiently handles large storage and computations involved.
%
%
In the past three decades a lot of work has been done to address dynamic graph problems in 
parallel~\cite{BorosEGK00,Ferragina95,OuR97,SchloegelKK02,SherlekarPR85}, 
semi-streaming \cite{AhnG13,AssadiKLY16,GuhaMT15,HuangP16,McGregor14}, 
and distributed (also called dynamic networks)~\cite{AwerbuchCK08,JinSLC14,KingKT15,LiuEK08,SwaminathanG98} 
environments.

In this paper, we address the problem of maintaining dynamic DFS tree efficiently in the parallel
environment and demonstrate its applications in semi-streaming and distributed environments.
  
\subsection{Existing results}
\label{sec:intro_exist_results}

In spite of the simplicity of a DFS tree, designing efficient parallel, 
distributed or streaming algorithms for a DFS tree has turned out to be quite challenging. 
Reif~\cite{Reif85} showed that the ordered DFS tree problem is a $P$-Complete problem. 
For many years, this result seemed to imply that the general DFS tree problem,
that is, the computation of any DFS tree of the graph is also inherently sequential. 
However, Aggarwal et al.~\cite{AggarwalA88,AggarwalAK90} proved that the general
DFS tree problem is in {\em RNC}
\footnote{{\em NC} is the class of problems solvable using $O(n^{c_1})$ processors
	in parallel in $O(\log^{c_2} n)$ time, for any constants $c_1$ and $c_2$.
	The class {\em RNC} extends {\em NC} to allow access to randomness. } 
 by designing a randomized EREW PRAM 
\footnote{Exclusive Read Exclusive Write (EREW) restricts any two processors 
	to simultaneously read or write the same memory cell.
	Concurrent Read Concurrent Write (CRCW) does not have this restriction.} 
algorithm that takes $\tilde{O}(1)$ time. 
But 
 the fastest deterministic algorithm for computing general DFS tree in parallel 
 still 
takes $\tilde{O}(\sqrt{n})$ time~\cite{AggarwalAK90,GoldbergPV93} in CRCW PRAM~\footnote{
It essentially shows DFS to be NC equivalent of minimum-weight perfect matching, 
which is in RNC whereas its best deterministic algorithm requires $\tilde{O}(\sqrt{n})$ time.
}
, even for undirected graphs.  
Moreover, the general DFS tree problem has been shown to be in {\em NC} for 
some special graphs including
DAGs~\cite{GhoshB84,Zhang86} 
and planar graphs~\cite{Hagerup90,Kao88,Smith86} (see~\cite{Freeman91} for a survey).
In fact for random graphs in $G(n,p)$ model
\footnote{$G(n,p)$ denotes a random graph where every edge of the graph exists independently with probability $p$.}
\cite{ErdosR59}, 
Dyer and Frieze~\cite{DyerF91a} proved that even ordered DFS tree problem is in {\em RNC}.
Whether general DFS tree problem is in {\em NC} is still a 
long standing open problem.

In the semi-streaming environment, the input graph is accessed in the form of a stream of graph edges,
where the algorithm is allowed only $O(n)$ local space. The DFS tree can be trivially computed 
using $O(n)$ passes over the input graph in the semi-streaming environment, 
each pass adding one vertex to the DFS tree. 
However, computing the DFS tree in $\tilde{O}(1)$ passes is considered hard~\cite{Farach-ColtonHL15}.
To the best of our knowledge, it remains an open problem to compute the DFS tree using even 
$o(n)$ passes in any relaxed streaming environment~\cite{ConnellC09,Ruhl03}.

Computing a DFS tree in a distributed setting was widely studied in 1980's and 1990's.
A DFS tree of the given graph can be computed in $O(n)$ rounds, with various trade offs 
of number of messages passed and size of each message. 
If the size of a message is allowed to be $O(n)$, the DFS tree can be built using $O(n)$ messages~\cite{KumarIS90,MakkiH96,SharmaI89}.
However, if the size of a message is limited to $\tilde{O}(1)$, 
the number of messages required is $O(m)$~\cite{Cidon88,LakshmananMT87,Tsin02}.

Thus, to maintain a DFS tree in dynamic setting, each update requires 
$\tilde{O}(\sqrt{n})$ time on a CRCW PRAM in deterministic parallel setting, 
$O(n)$ passes in the semi-streaming setting 
and $O(n)$ rounds in the distributed setting, which is very inefficient.
Hence, exploring the dynamic maintenance of a DFS tree in parallel, semi-streaming and distributed  
environments seems to be a long neglected problem of practical significance.

We present optimal algorithms (up to $poly\log n$ factors) for maintaining a fully dynamic DFS tree
for an undirected graph under both edge and vertex updates on these models. 

\subsection{Our Results}
\label{sec:results}
We consider an extended notion of updates wherein an update could be either
insertion/deletion of a vertex or insertion/deletion of an edge. 
Furthermore, an inserted vertex can be added with any set of incident edges to the graph.

In the parallel setting, our main result can be succinctly described as follows. 
\begin{theorem} Given an undirected graph and its DFS tree, it can be preprocessed to
	 build a data structure of size $O(m)$ in $O(\log n)$ time using $m$ processors on an EREW PRAM 
	 such that for any update in the graph, a DFS tree of the updated graph can be computed in 
	 $O(\log^3 n)$ time using $n$ processors on an EREW PRAM.
	\label{main-result}
\end{theorem}

With this result at the core, we easily obtain the following results.

\begin{enumerate}
	\item {\em Parallel Fully Dynamic DFS}:	\\
	Given any arbitrary online sequence of vertex or edge updates, 
	we can maintain a DFS tree of an undirected graph in $O(\log^3 n)$ time per update 
	using $m$ processors on an EREW PRAM. 

	\item {\em Parallel Fault tolerant DFS}:\\	
	An undirected graph can be preprocessed to build a data structure of size $O(m)$ such that 
	for any set of $k (\leq \log n)$ updates in the graph, a DFS tree of the updated graph 
	can be computed in $O(k\log^{2k+1} n)$ 
	time using $n$ processors on an EREW PRAM.

\end{enumerate}  

Our fully dynamic algorithm and fault tolerant algorithm (for constant $k$), 
clearly take optimal time (up to $poly\log n$ factors) for maintaining a DFS tree.
Our fault tolerant algorithm (for constant $k$) is also work optimal
(upto $poly\log n$ factors) since a single update can lead to $\Theta(n)$ changes in the DFS tree.
Moreover, our result also establishes that maintaining a fully dynamic DFS tree for an undirected graph is in {\em NC} 
(which is still an open problem for DFS tree in the static setting).

\subsection{Applications of Parallel Fully Dynamic DFS}
Our parallel fully dynamic DFS algorithm can be seamlessly adapted to the semi-streaming and distributed environments
as follows.

%
%
%

\begin{enumerate}
	\item {\em Semi-streaming Fully Dynamic DFS}: \\
	Given any arbitrary online sequence of vertex or edge updates, we can maintain a DFS tree of an undirected graph 
	using $O(\log^2 n)$ passes over the input graph per update 
	by a semi-streaming algorithm using $O(n)$ space.

	\item {\em Distributed Fully Dynamic DFS}: \\
	Given any arbitrary online sequence of vertex or edge updates,  we can maintain a DFS tree of an undirected graph 
	in $O(D\log^2 n)$ rounds per update in the synchronous $\mathcal{CONGEST}(n/D)$ model 
	\footnote{${\mathcal CONGEST}(B)$ model is the standard $\mathcal{CONGEST}$ model~\cite{Peleg00} where message size is relaxed to $B$ words.}
	using $O(nD\log^2 n+m)$ messages of size $O(n/D)$ requiring $O(n)$ space on each processor, where $D$ is diameter of the graph. 


\end{enumerate}  

Our semi-streaming algorithm clearly takes optimal number of passes 
(up to $poly\log n$ factors) for maintaining a DFS tree.
Our distributed algorithm that works in a restricted ${\mathcal CONGEST}(B)$ model, 
also arguably requires optimal rounds (up to $poly\log n$ factors) because it requires 
$\Omega(D)$ rounds to propagate the information of the update throughout the graph. 
Since almost the whole DFS tree may need to be updated due to a single 
update in the graph, every algorithm for maintaining a DFS tree in the distributed setting will require $\Omega(D)$ rounds
\footnote{For an algorithm maintaining the whole DFS tree at each node, even our message size is optimal.
	This is because an update of size $O(n)$ (vertex insertion with arbitrary set of edges) 
	will have to be propagated throughout the network in the worst case.
	In $O(D)$ rounds, it can only be propagated using messages of size $\Omega(n/D)$.
	(see Section~\ref{sec:distributed} for details).}.
This essentially improves the state of the art for the classes of graphs with $o(n)$ diameter.

\subsection{Overview}
\label{sec:overview}
We now describe a brief overview of our result. 
Baswana et al.~\cite{BaswanaCCK15} proved that updating a DFS tree after any update in the graph 
is equivalent to {\em rerooting} disjoint subtrees of the DFS tree.
They also presented an algorithm to reroot a DFS tree $T$ (or its subtree), originally rooted at $r$
to a new root $r'$, in $\tilde{O}(n)$ time.
It starts the traversal from $r'$ traversing the path connecting $r'$ to $r$ in $T$. 
Now, the subtrees hanging from this path are essentially the components of the \textit{unvisited graph} 
(the subgraph induced by the unvisited vertices of the graph) due to the absence of {\em cross edges}. 
In the updated DFS tree, every such subtree, say $\tau$, shall hang from an edge emanating from $\tau$ 
to the path from $r'$ to $r$. Let this edge be $(x,y)$, where $x\in \tau$. 
Thus, we need to recursively reroot $\tau$ to the new root $x$ and hang it from $(x,y)$ 
in the updated DFS tree. Note that this rerooting can be independently performed for different 
subtrees hanging from tree path from $r'$ to $r$. 

At the core of their result, they use a property of the DFS tree, that they called {\em components} property, 
to find the edge $(x,y)$ efficiently, using a data structure ${\mathcal D}_0$.
However, as evident from the discussion above, their rerooting procedure can be strictly sequential in the worst case.
This is because the size of a subtree $\tau$ to be rerooted can be almost equal to that of the original tree $T$.
As a result, $O(n)$ sequential reroots may be required in the worst case. 
Our main contribution is an algorithm that performs this rerooting 
efficiently in parallel.

Our algorithm ensures that rerooting is completed in $\tilde{O}(1)$ steps as follows. 
At any point of time, we ensure that every component $c$ of the {\em unvisited graph} 
is either of type $C1$, having a single subtree of $T$, or of type $C2$, having a path $p_c$ 
and a set of subtrees of $T$ having edges to $p_c$. Note that in~\cite{BaswanaCCK15} every component 
of the unvisited graph is of type $C1$.
We define three types of traversals, namely, {\em path halving} (also used by~\cite{BaswanaCCK15}), 
{\em disintegrating traversal} and {\em disconnecting traversal}.
We prove that using a combination of $O(1)$ such traversals, for every component $c$ of the unvisited graph,
either the length of $p_c$ is halved or the size of largest subtree in $c$ is halved.
%
Moreover, these traversals can be performed in $O(\log n)$ time on $|c|$ processors
using the {\em components} property and a data structure $\mathcal D$ 
(answering similar queries as ${\mathcal D}_0$). 
However, since our algorithm ensures that each vertex is queried by $\mathcal D$ only $\tilde{O}(1)$ times 
(unlike~\cite{BaswanaCCK15}), our data structure $\mathcal D$ is much simpler than ${\mathcal D}_0$.



Furthermore, both our algorithm and the algorithm by~\cite{BaswanaCCK15} use the non-tree edges of the graph only to 
answer queries on data structure $\mathcal D$ (or ${\mathcal D}_0$).
The remaining operations (except for queries on $\mathcal D$) required by our algorithm can be 
performed using only edges of $T$ in $O(n)$ space. 
As a result, our algorithm being efficient in parallel setting (unlike~\cite{BaswanaCCK15}), 
can also be adapted to the semi-streaming  and distributed model as follows.
In the semi-streaming model, the passes over the input graph are used only to answer the queries on $\mathcal D$,
where the parallel queries on $\mathcal D$ made by our algorithm can be answered simultaneously using a single pass. 
Our distributed algorithm only needs to store the current DFS tree at each node and 
the adjacency list of the corresponding vertex abiding the restriction of $O(n)$ space at each node. 
Again, the distributed computation is only used to answer queries on $\mathcal D$. 

\section{Preliminaries}
\label{sec:prelim}

Let $G=(V,E)$ be any given undirected graph on $n=|V|$ vertices and $m=|E|$ edges. 
The following notations will be used throughout the paper.

\begin{itemize}
\item $par(w):$~ Parent of $w$ in $T$.  
\item  $T(x):$ The subtree of $T$ rooted at vertex $x$. 
\item  $path(x,y):$ Path from vertex $x$ to vertex $y$ in $T$.
\item  $LCA(x,y):$ Lowest common ancestor of $x$ and $y$ in $T$.
\item $root(T'):$~ Root of a subtree $T'$ of $T$, i.e., $root\big(T(x)\big)=x$.
\item $T^*:$~ The DFS tree computed by our algorithm for the updated graph.
\end{itemize}

A subtree $T'$ is said to be {\em hanging} from a path $p$ if the $root(T')$ is a child of 
some vertex on the path $p$ and does not belong to the path $p$. 
Unless stated otherwise, a component refers to a connected component of the unvisited graph.
We refer to a path $p$ in a DFS tree $T$ as an {\em ancestor-descendant} path if one of its endpoints 
is an ancestor of the other in $T$.

For our distributed algorithm, we use the synchronous ${\cal CONGEST}(B)$ model~\cite{Peleg00}. 
For the dynamic setting, Henzinger et al.~\cite{HenzingerKN13} presented a model 
that has a {\em preprocessing} stage followed by an alternating sequence of non-overlapping stages for 
{\em update} and {\em recovery} (see Section~\ref{sec:distributed} for details). 
We use this model with an additional constraint of space restriction of $O(n)$ size at each node.
In the absence of this restriction, the whole graph can be stored at each node, 
where an algorithm can trivially propagate the update to each node and the updated solution can be computed locally.
Also, we allow the deletion updates to be {\em abrupt}, i.e., the deleted link/node becomes unavailable 
for use instantly after the update. 

In order to handle disconnected graphs, we add a dummy vertex $r$ to the graph
and connect it to all the vertices. Our algorithm maintains a DFS tree rooted at 
$r$ in this augmented graph, where each child subtree of $r$ is a DFS tree of a connected component 
in the DFS forest of the original graph. 

We shall now define some queries that are performed by our algorithm on the data structure $\cal D$
(similar queries on ${\cal D}_0$ also used in \cite{BaswanaCCK15}).
Let $v,w,x,y\in V$, where $path(x,y)$ and $path(v,w)$ (if required) are ancestor-descendant paths in $T$.
Also, no vertex in $path(v,w)$ is a descendant of any vertex in $path(x,y)$.
We define the following queries.

\begin{enumerate}
\item $Query\big(w,path(x,y)\big):$ among all the edges from $w$ that are incident on $path(x,y)$ 
					 in $G$, return an edge that is incident nearest to $x$ on $path(x,y)$.%
\item $Query\big(T(w),path(x,y)\big):$ among all the edges from $T(w)$ that are incident on $path(x,y)$ 
					in $G$, return an edge that is incident nearest to $x$ on $path(x,y)$.
\item $Query\big(path(v,w),path(x,y)\big):$ among all the edges from $path(v,w)$ that are incident on 
					 $path(x,y)$ in $G$, return an edge that is incident nearest to $x$ on $path(x,y)$.%
\end{enumerate}

Let the  {\em descendant} vertices of the three queries described above be $w$, $T(w)$ and $path(v,w)$ respectively.
A set of queries on the data structure $\cal D$ are called {\em independent} if the 
{\em descendant} vertices of these queries are disjoint.



Baswana et al.~\cite{BaswanaCCK15} described the {\em components} property of a DFS tree as follows.


\begin{figure}[!ht]
\centering
\includegraphics[width=0.35\linewidth]{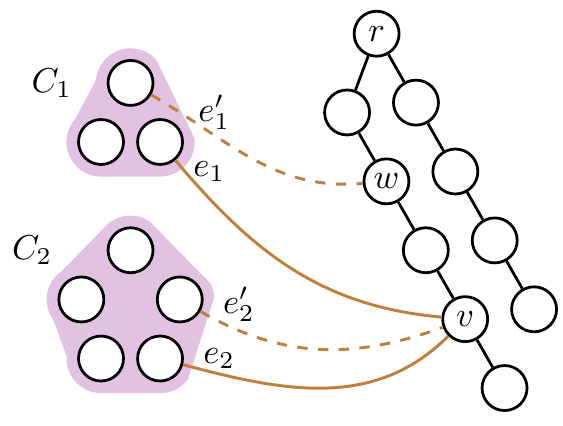}
\caption{Edges $e_1'$ and $e_2'$ can be ignored during the DFS traversal (reproduced from~\cite{BaswanaCCK16}).}
\label{fig:component-property}
\end{figure}

\begin{lemma}[Components Property~\cite{BaswanaCCK15}]
\label{lemma:L()}
Let $T^*$ be the partially built DFS tree and $v$ be the vertex currently being visited. 
Let $C_1,..,C_k$ be the connected components of the subgraph induced by the unvisited vertices.  
For any two edges $e_i$ and $e'_i$ from $C_i$ that are incident respectively on $v$ and some 
ancestor (not necessarily proper) $w$ of $v$ in $T^*$, it is sufficient to 
consider only $e_i$ during the DFS traversal, i.e., the edge $e'_i$ can be safely ignored. 
\end{lemma}


Ignoring $e'_i$ during the DFS traversal, as stated in the components property, is justified 
because $e'_i$ will appear as a back edge in the resulting DFS tree (refer to Figure~\ref{fig:component-property}).
The edge $e_i$ can be found by querying the data structure $\cal D$ (or ${\cal D}_0$ in \cite{BaswanaCCK15}).
The DFS tree is then updated after any update in the graph by reducing it to {\em rerooting} disjoint subtrees of the DFS tree 
using the components property. Rerooting a subtree $T(v)$ at a new root $r'\in T(v)$ involves restructuring the tree $T(v)$
to be now rooted at $r'$ such that the new tree is also a DFS tree of the subgraph induced by $T(v)$.
This reduction will henceforth be referred as the {\em reduction} algorithm and is described in the following 
section. 


\section{Reduction Algorithm}
\label{sec:OverviewApp}
We now describe how updating a DFS tree after any kind of update in the graph is equivalent 
to a simple procedure, i.e., {\em rerooting} disjoint subtrees of the DFS tree. 
Note that similar reduction was also used by Baswana et al.~\cite{BaswanaCCK15} but we describe it here for the sake of completeness
as follows (see Figure~\ref{figure:overview-updates}).

\begin{figure*}[!ht]
\centering
\includegraphics[width=\linewidth]{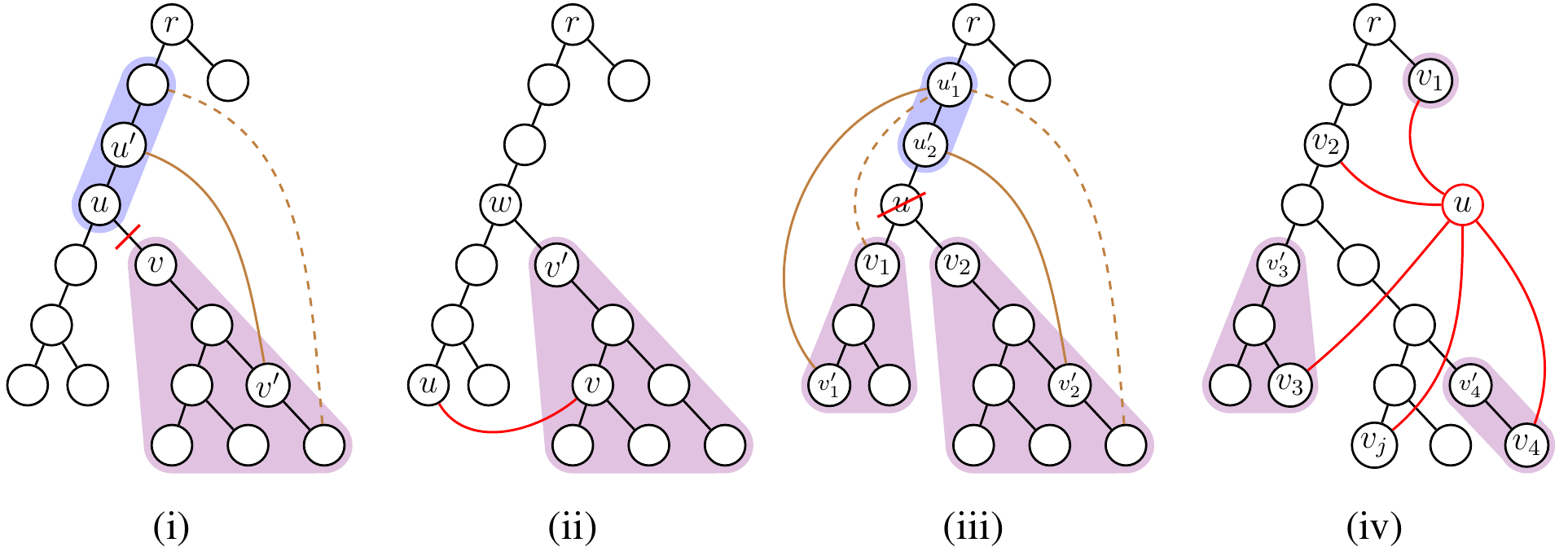}
\caption{Updating the DFS tree after a single update: (i) deletion of an edge, (ii) insertion of an edge, 
(iii) deletion of a vertex, and (iv) insertion of a vertex. 
The algorithm reroots the marked subtrees (marked in violet) and hangs it from the inserted edge (in case of insertion)
or the lowest edge (in case of deletion) on the marked path (marked in blue) from the marked subtree.
(reproduced from~\cite{BaswanaCCK15})
}
\label{figure:overview-updates}
\end{figure*}

\begin{enumerate}
\item \textbf{Deletion of an edge $(u,v)$:}\\
If $(u,v)$ is a back edge in $T$, simply delete it from the graph.
Otherwise, let $u = par(v)$ in $T$. 
The algorithm finds the lowest edge $(u',v')$ on the $path(u,r)$ from $T(v)$, where $v'\in T(v)$. 
The subtree $T(v)$ can then be rerooted to the new root $v'$ and hanged from $u'$ using $(u',v')$ 
to get the final tree $T^*$.
\item \textbf{Insertion of an edge $(u,v)$:}\\
In case $(u,v)$ is a back edge, simply insert it in the graph.
Otherwise, let $w$ be the LCA of $u$ and $v$ in $T$ and $v'$ be the child of $w$ such that $v\in T(v')$.
The subtree $T(v')$ can then be rerooted to the new root $v$ and hanged from $u$ using $(u,v)$ to get the 
final tree $T^*$.
\item \textbf{Deletion of a vertex $u$:}\\
Let $v_1,...,v_c$ be the children of $u$ in $T$. 
For each subtree $T(v_i)$, the algorithm finds the lowest edge $(u'_i,v'_i)$ on the $path(par(u),r)$ 
from $T(v_i)$, where $v'_i\in T(v_i)$. Each subtree $T(v_i)$ can then be rerooted to the new root $v'_i$ 
and hanged from $u'_i$ using $(u'_i,v'_i)$ to get the final tree $T^*$.
\item \textbf{Insertion of a vertex $u$:}\\
Let $v_1,...,v_c$ be the neighbors of $u$ in the graph. 
Arbitrarily choose a neighbor $v_j$ and make $u$ the child of $v_j$ in $T^*$. 
For each $v_i$, such that $v_i\notin path(v_j,r)$, let $T(v'_i)$ be the subtree 
hanging from $path(v_j,r)$ such that $v_i\in T(v'_i)$. 
Each subtree $T(v'_i)$ can then be rerooted to the new root $v_i$ and 
hanged from $u$ using $(u,v_i)$ to get the final tree $T^*$.
\end{enumerate}

In case of a vertex update, multiple subtrees may be required to be rerooted by the algorithm. 
Let these subtrees be $T_1,...,T_c$. Notice that each of these subtrees can be rerooted
independent of each other, and hence in parallel. However, in order to perform the 
{\em reduction} algorithm efficiently in parallel, 
we require a structure to answer the following queries efficiently in parallel. (a) Finding LCA of two vertices in $T$. 
(b) Finding the highest edge from a subtree $T(v)$ to a path in $T$ (a query on data structure $\cal D$). 
In addition to these we also require several other types of queries to be efficiently answered in parallel setting as 
testing if an edge is back edge, finding vertices on a path, child subtree of a vertex containing a given vertex etc. 
However, these can easily be answered using LCA queries as described in Section~\ref{sec:ds}. 
Thus, we have the following theorem.


\begin{theorem}
Given an undirected graph $G$ and its DFS tree $T$, any graph update can be reduced to 
independently rerooting disjoint subtrees of $T$ by performing $O(1)$ sets of independent queries on the 
data structure $\cal D$ and $O(1)$ sets of LCA queries on $T$, where each set has at most $n$ queries.
\label{thm:convertApp}
\end{theorem}

\begin{remark}
The implementation of reduction algorithm is simpler in distributed and semi-streaming environments, 
where any operation on the DFS tree $T$ can be performed locally without any distributed computation or 
passes over the input graph respectively. Hence, for these environments the reduction algorithm requires 
only $O(1)$ sets of independent queries on the data structure $\cal D$.
\end{remark}


\section{Rerooting a DFS tree}
\label{sec:pdfs}
We now describe the algorithm to reroot a subtree  $T(r_0)$ of the DFS tree $T$, 
from its original root $r_0$ to the new root $r^*$. 
Also, let the data structure $\cal D$ be built on $T$ (see Section~\ref{sec:prelim}).
Also, we maintain the following invariant: at any moment of the algorithm, 
every component $c$ of the unvisited graph can be of the following two types: 

\begin{itemize}
\item[C1:] Consists of a single subtree $\tau_c$ of the DFS tree $T$. 
\item[C2:] Consists of a single {\em ancestor-descendant} path $p_c$ and a set 
			${\cal T}_c$ of subtrees of the DFS tree $T$ having at least one edge to $p_c$. 
			Note that for any $\tau_1,\tau_2\in{\cal T}_c$, there is no edge between $\tau_1$ and $\tau_2$ since $T$ is a DFS tree.
\end{itemize}
Moreover, for every component $c$ we also have a vertex $r_c\in c$ from which the DFS tree of the component $c$ 
would be rooted in the final DFS tree $T^*$.

\newcommand{\Th}{{\mathbb{T}}}

The algorithm is divided into $\log n$ {\em phases}, where each phase is further divided 
into $\log n$ {\em stages}.
At the end of phase $\PP_i$, every subtree of any component $c$ ($\tau_c$ or subtrees in ${\cal T}_c$) 
has at most $n/2^i$ vertices.
During phase $\PP_i$, every component has at least one {\em heavy} subtree (having $> n/2^{i}$ vertices). 
If no such tree exists, we move the component to the next phase. We denote the set of these heavy subtrees by $\Th_c$.
For notational convenience, we refer to the heaviest subtree of every component $c$ as 
$\tau_c$, even for components of type $C2$.
Hence, for any component of type $C1$ or $C2$, we have $\tau_c\in \Th_c$. 
Clearly the algorithm ends after $\log n$ phases as every component of the unvisited graph would be empty.  

At the end of stage $\St_j$ of a phase, the length of $p_c$ in each component $c$ is at most $n/2^j$. 
If $|p_c|\leq n/2^{j}$, we move the component $c$ to the next stage.
Further, for any component $c$ of type $C1$, the value of $|p_c|$ is zero, 
so we move such components to the last stage of the phase, i.e., $\St_{\log n}$. 
Clearly at the end of $\log n$ stages, each component would be of type $C1$.

In the beginning of the algorithm, we have the component induced by $T(r_0)$ of type $C1$ where $r_c=r^*$.
Note that during each stage, different connected components of the unvisited graph can be 
processed independent of each other in parallel.
 
\begin{figure*}[!ht]
\centering
\includegraphics[width=.85\linewidth]{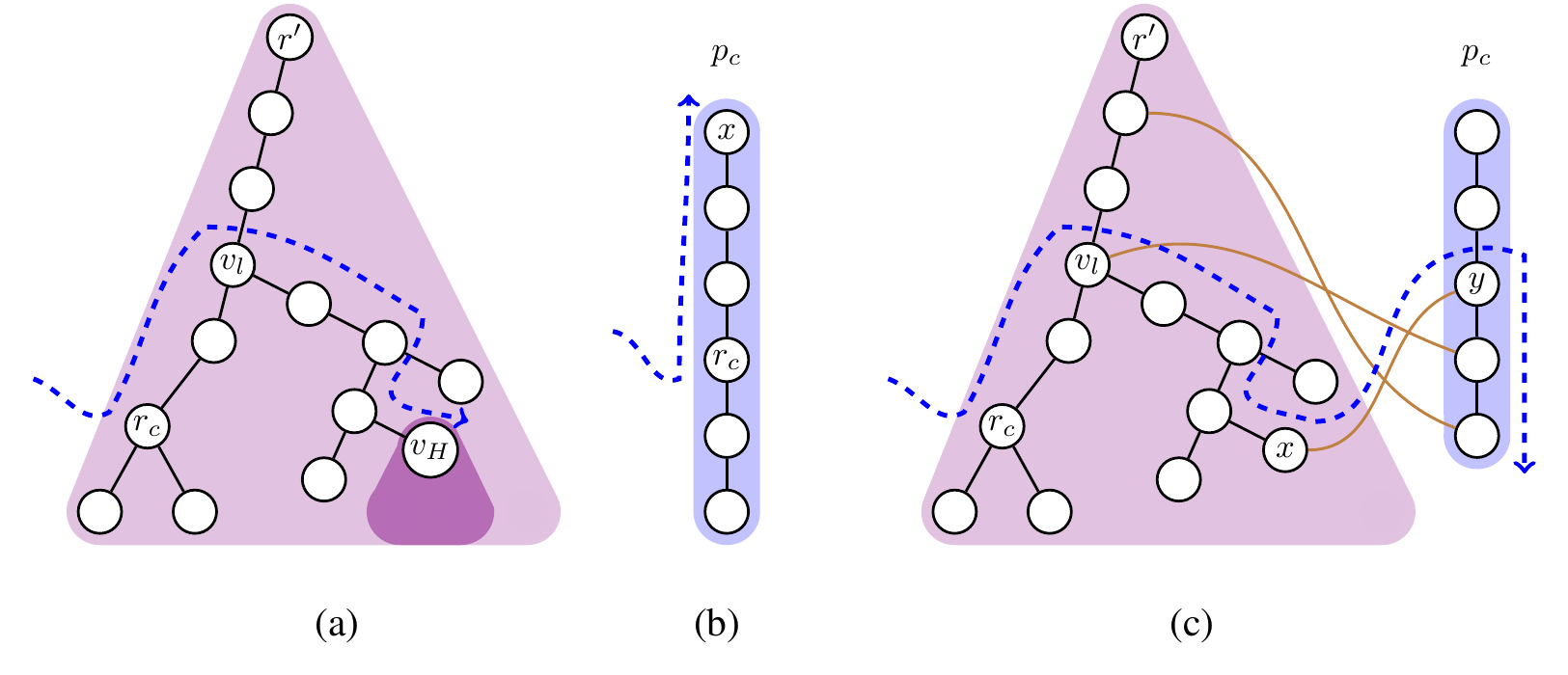}
\caption{The three easier types of traversals shown in blue dotted lines,
(a) Disintegrating traversal, (b) Path Halving and (c) Disconnecting traversal.}
\label{fig:pdfs1}
\end{figure*}

\subsection*{Algorithm}
We now describe how a component $c$ in phase $\PP_i$ and stage $\St_j$ is traversed by our algorithm.
The aim is to build a partial DFS tree for the component $c$ rooted at $r_c$, 
that can be attached to the partially built DFS tree $T^*$ of the updated graph.
Note that this has to be performed in such a manner that every component of the unvisited part of $c$
is of type $C1$ or $C2$ only.

Now, in order to move to the next phase, we need to ensure that for every component $c'$ of the unvisited part of $c$, 
$|\tau_{c'}|\leq n/2^i$.  As described above, after $\log n$ stages every component $c'$ is of type $C1$.
Thus, we perform a {\em disintegrating traversal} of $\tau_c$ which ensures
that every component of the unvisited part of $c$ can be moved to the next phase.

During $\St_j$, in order to move to the next stage, we need to ensure that for every component $c'$ of the unvisited part of $c$,
either $|p_{c'}|\leq n/2^j$ (moving it to next stage) or  $|\tau_{c'}|\leq n/2^i$ (moving it to next phase). 
The component is processed based on the location of $r_c$ in $c$ as follows.
If $r_c\in p_c$, we perform {\em path halving} which ensures the components move to the next stage.
If $r_c\in \tau\notin \Th_c$, we perform a {\em disconnecting traversal} of $\tau$ followed by {\em path halving} of $p_c$
such that the unvisited components of $\tau$ are no longer connected to residual part of $p_c$,
 moving them to the next phase.
The remaining components of $c$ moves to the next stage due to path halving.

We shall refer to disintegrating traversal, path halving and disconnecting traversal as
the {\em simpler} traversals. The difficult case is when $r_c\in \tau\in \Th_c$. 
Here, some trivial cases can be directly processed by the {\em simpler} traversals mentioned above. 
For the remaining cases we perform {\em heavy subtree traversal} of $\tau$ which shall ensure that 
the unvisited part of $c$ reduces to those requiring {\em simpler} traversals. 
Refer to Procedure~\ref{alg:parallel_dfs} in Appendix~\ref{appn:pseudocode} 
for the pseudo-code of the main algorithm.

We now describe the different types of traversals in detail.
For any component $c$, we refer to the smallest subtree of $\tau\in \Th_c$ that has more than 
$n/2^i$ vertices as $T(v_H)$. Since $n/2^{i-1}\geq |\tau|> n/2^i$, $v_H$ is unique.
Also, let $r'=root(\tau)$ (if $r_c\in \tau$) and $v_l=LCA(r_c,v_H)$.

\subsection{Disintegrating Traversal}
\label{sec:c1}
Consider a component $c$ of type $C1$ with new root $r_c\in \tau_c$
in phase $\PP_i$  ($n/2^i<|\tau_c|\leq n/2^{i-1}$). We first find the vertex $v_H$.
We then traverse along the tree path $path(r_c,v_H)$, adding it to $T^*$ 
\big(see Figure~\ref{fig:pdfs1} (a)\big). 
Now, the unvisited part of $c$ consists of $path(par(v_l),r')$ (say $p$) 
and the subtrees hanging from $path(r_c,r')$ and $path(v_l,v_H)$.
Notice that $p$ is an ancestor-descendant path of $T$ and each subtree has at most $n/2^i$ vertices. 
Each subtree not having an edge to $p$ corresponds to a separate component of type $C1$. 
The path $p$ and the remaining subtrees (having an edge to $p$) form a component of type $C2$. 
For each component $c^*$, we also need to find the new root $r_{c^*}$ for the updated DFS tree of the component.
Using the components property, we know $r_{c^*}$ has the lowest edge from $c^*$ on the path $p^*$,
where $p^*$ is the newly attached path to $T^*$ described above. 
Both these queries (finding an edge to $p$ and the lowest edge on $p^*$) can be answered by our data structure $\cal D$
(see Section~\ref{sec:prelim}). Thus, every component $c^*$ can be identified and moved to next phase. 
Refer to Procedure~\ref{alg:dfs_c1} in Appendix~\ref{appn:pseudocode} for the pseudo code.

\begin{remark}
If $r_c = r'$, this traversal can also be performed on a subtree 
from a component $c$ of type $C2$ achieving similar result. 
This is possible because no new path $p$ would be formed and we 
still get components of type $C1$ and $C2$ (being connected to a single path $p_c$).
\end{remark}

\subsection{Path Halving}
Consider a component of type $C2$ with $r_c\in p_c=path(x,y)$.
We first find the farther end of $p_c$, say $x$, where $|path(r_c,x)|\geq |path(r_c,y)|$.
We then traverse from $r_c$ to $x$ adding $path(r_c,x)$ to the tree $T^*$ \big(see Figure~\ref{fig:pdfs1} (b)\big).
The component $c'$ of type $C2$ thus formed will have $p_{c'}$ of length at most half of $p_c$. 
Now, the subtrees in $c$ having an edge to $p_{c'}$ would be a part of $c'$.
The remaining subtrees would form individual components of type $C1$. 
Again, the new root of each component can be found using $\cal D$
by querying for the lowest edge on the $path(r_c,x)$ added to $T^*$.
Refer to Procedure ~\ref{alg:dfs_c2_3} in Appendix~\ref{appn:pseudocode} for the pseudo code.

\subsection{Disconnecting Traversal}
\label{sec:discon}
Consider a component of type $C2$ with $r_c\in \tau$, where $\tau\notin \Th_c$. 
We traverse $\tau$ from $r_c$ to reach $p_c$, which is then followed by path halving of $p_c$.
The goal is to ensure that the unvisited part of $\tau$ is not connected to the 
unvisited part of $p_c$ (say $p'$) after path halving, moving it to the next phase.
The remaining subtrees of $c$ with $p'$ will move to the next stage as a result of path halving of $p_c$.

Now, if at least one edge from $\tau$ is present on the upper half of $p_c$,
we find the highest edge from $\tau$ to $p_c$ \big(see Figure~\ref{fig:pdfs1} (c)\big).
Otherwise, we find the lowest edge from $\tau$ to $p_c$.
Let it be $(x,y)$, where $y\in p_c$ and $x\in \tau$.
This ensures that on entering $p_c$ through $y$, path halving would ensure that all the edges from $\tau$ to $p_c$
are incident on the traversed part of $p_c$ (say $p$). 

We perform the traversal from $r_c$ to $x$ similar to the {\em disintegrating traversal} along 
$path(r_c,x)$, attaching it to $T^*$. 
Since none of the components of unvisited part of $\tau$ are connected to $p'$, 
all the components formed would be of type $C1$ or $C2$ as described in Section~\ref{sec:c1}. 
However, while finding the new root of each resulting component $c'$, 
we also need to consider the lowest edge from the component on $p$.
Further, since $\tau\notin \Th_c$, size of each subtree in the resulting components is at most $n/2^i$.
Thus, the resultant components of $\tau$ are moved to the next phase
(see Procedure~\ref{alg:dfs_c2_2} in Appendix~\ref{appn:pseudocode} for pseudo code).

\begin{remark} 
If $r_c\in T(v_H)$, this traversal can also be performed on a $\tau\in \Th_c$ getting a similar result.
This is because each subtree in resultant components of $\tau$ will have size at most $n/2^i$
moving it to the next phase. 
However, if $r_c\notin T(v_H)$ the resultant component $c'$ of type $C2$ formed can
have a heavy subtree and a path $p_{c'}$ of arbitrary length.
This is not permitted as it will move the component to some earlier stage in the same phase.
\end{remark}

\setcounter{scenario}{0}

\begin{figure*}[!ht]
	\centering
	\includegraphics[width=\linewidth]{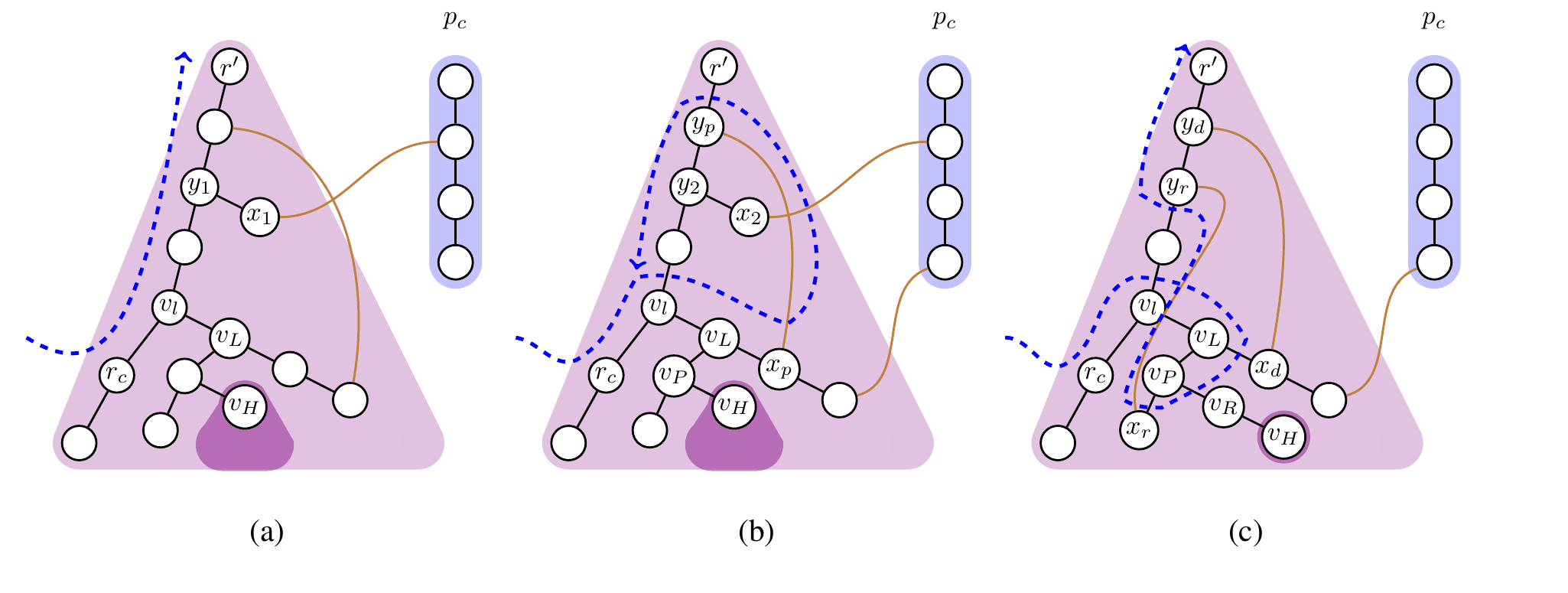}
	\caption{The three scenarios for Heavy Subtree Traversal (blue dotted lines) showing (a) $l$ traversal, 
		(b) $p$ traversal and (c) $r$ traversal.
	}
	\label{fig:pdfs}
\end{figure*}

\subsection{Heavy Subtree Traversal}
\label{sec:heavyP}
Consider a component $c$ of type $C2$ with $r_c\in\tau$, where $\tau \in \Th_c$.
As described earlier, if $r_c=root(\tau)$ or $r_c\in T(v_H)$, 
the heavy subtree $\tau$ can be processed using 
disintegrating or disconnecting traversals respectively. 
Otherwise, we traverse it using one of three scenarios. Our algorithm checks each scenario 
in turn for its applicability to $\tau$, 
eventually choosing a scenario to perform an $l,p$ or $r$ traversal (see Figure~\ref{fig:pdfs}). 
This traversal ensures that it is followed by a {\em simpler} traversal described earlier,
so that each component will either move to the next phase or the next stage.
We shall also prove that these scenarios are indeed exhaustive, i.e., for any $\tau$, 
one of the scenarios is indeed applicable ensuring that each component moves to the next stage/phase. 
The following lemma describes the conditions when a scenario is applicable. 


\newcommand{\Aa}{${\cal A}_1$}
\newcommand{\Ab}{${\cal A}_2$}
\newcommand{\Ac}{${\cal A}_3$}

\begin{lemma}[Applicability Lemma]
After a traversal of path $p^*$ in a subtree $\tau\in \Th_c$, 
every component of unvisited part of $c$ can be moved to the next phase/stage 
using a {\em simpler} traversal if
\begin{enumerate}
\item[\Aa:] Traversal of $p^*$ produces components of type $C1$ or $C2$ only, 
\item[\Ab:] The subtree $T(v_H)$ is connected to $p_c$ (if in component of type $C2$), 
\item[\Ac:] The lowest edge on $p^*$ from the component containing $p_c$ is not a 
back edge from the subtree containing $T(v_H)$ with its end point outside $T(v_H)$.
\end{enumerate}
\label{lem:non-applicable}
\end{lemma}
\begin{proof}
	Consider any traversal satisfying the above criteria, which forms components of type $C1$ and $C2$ only.
	For each such component $c'$, we find the lowest edge $e'$ from $c'$ to the traversed path, giving the new root $r_{c'}$.
	Every component which does not contain $p_c$ or $T(v_H)$ can be directly 
	moved to the next phase with root $r_{c'}$, because the remaining subtrees of $\tau$ 
	(not containing $T(v_H)$) cannot be heavy. 
	In case the component containing $T(v_H)$ is of type $C1$ it can be moved to the last stage of the phase. 
	In case the component $c'$ containing $p_c$ does not contain $T(v_H)$, we have $r_{c'}\in p_c$ 
	or $r_{c'}\in\tau'$ (a non-heavy subtree of $\tau$), moving $c'$ to the next stage after performing 
	{\em path halving} or {\em disconnecting} traversal of $\tau'$ respectively.
	Due to the second condition, this only leaves the component $c'$ of type $C2$ having both $p_c$ and 
	a subtree $T(v_h)\in \TT_{c'}$ which contains $T(v_H)$. 
	The third condition prevents $e'$ from being a back edge with $r_{c'}\in T(v_h)$ and $r_{c'}\notin T(v_H)$.
	The remaining cases can be processed by {\em path halving} (if $r_{c'}\in p_c$), 
	{\em disconnecting} traversal ($r_{c'}\notin T(v_h)$ or $r_{c'}\in T(v_H)$) or {\em disintegrating} traversal 
	($r_{c'}=v_h$, includes $e'$ being a tree edge) respectively.
\end{proof}

\begin{remark}
Applicability lemma is employed when $p^*$ does not traverse through $T(v_H)$, 
in which case the unvisited component trivially moves to the next stage/phase.
In such a case, the traversal preceding $p^*$ was clearly applicable.
Condition \Ab~ ensures that the heavy subtree containing $T(v_H)$
does not form a component $c'$ with arbitrary length of path $p_{c'}$, as this can move it to 
some previous stage which is not allowed. Because of the same reason {\em disconnecting} traversal is not
used on such heavy subtrees. 
\end{remark}

We now briefly describe the three scenarios, namely, $l,p$ and $r$ traversals and define 
a few notations related to them (shown in Figure~\ref{fig:pdfs}).
The $l,p$ and $r$ traversals follow the path shown in figure (using blue dotted lines) 
which shall henceforth be referred as $p^*_L$, $p^*_P$ and $p^*_R$ respectively.
Both $p$ and $r$ traversals use a back edge during the traversal, 
denoted by $(x_p,y_p)$ and $(x_r,y_r)$ respectively. 
Further, we refer to the subtrees containing $v_H$ that hangs from $p^*_L$, $p^*_P$ and $p^*_R$ as 
$T(v_L)$, $T(v_P)$ and $T(v_R)$ respectively. 
We shall refer to the subtree hanging from the traversed path ($p^*_L$, $p^*_P$ or $p^*_R$) 
with an edge to $p_c$ as {\em eligible} subtrees.
In each scenario we ensure \Aa~ and \Ab~ 
by construction, 
implying that the scenario will not be applicable only if the third condition is violated.
Thus, we only need to find the lowest edge on traversed path from the {\em eligible} subtrees, 
to determine the applicability of a scenario.
Also, the edges $(x_p,y_p)$ and $(x_r,y_r)$ are chosen in such a way that if $l$ and $p$ traversals are not applicable, 
then $r$ traversal always satisfies {\em applicability} lemma, with the lowest edge from component containing $p_c$ 
being $(x_d,y_d)$, where $x_d\in \tau_d \neq T(v_R)$.

\subsubsection*{Scenario 1: {\large $l$} traversal}
Consider the traversal shown in Figure~\ref{fig:pdfs}~(a), where $p^*_L=path(r_c,r')$. 
Since, this traversal does not create a new non-traversed path, 
the first two conditions of {\em applicability} lemma are satisfied.
We find the lowest edge on $p^*_L$ (highest edge on $path(r_c,r')$) 
from $p_c$ and the {\em eligible} subtrees, say $(x_1,y_1)$, where $y_1\in p^*_L$. 
In case this edge satisfies the third condition of {\em applicability} lemma,
we perform the traversal otherwise move to the next scenario.

\begin{remark}
This scenario is not applicable only if $(x_1,y_1)$ is a back edge with $x_1\in T(v_L)$ 
and $x_1\notin T(v_H)$.
\end{remark}

\subsubsection*{Scenario 2: {\large $p$} traversal}
Consider the traversal shown in Figure~\ref{fig:pdfs}~(b), where 
$p^*_P=path(r_c,x_p)\cup (x_p,y_p) \cup path(y_p,par(v_l))$.
To perform this traversal, we choose $(x_p,y_p)$ along with $(x_d,y_d)$ such that
either $p$ traversal is applicable, or $r$ traversal is applicable using $(x_d,y_d)$.
We now describe how such a choice of $(x_p,y_p)$  and $(x_d,y_d)$ can be made.
Let $\tau_d$ and $\tau_p$ (if any), be the subtrees hanging from $path(v_L,v_H)$ 
containing $x_d$ and $x_p$ respectively.

\subsubsection*{Choice of {\large $(x_p,y_p)$}  and {\large $(x_d,y_d)$} }
\label{sec:choose_ep_ed}
We find the highest edge on $path(r_c,r')$ from the {\em eligible} subtrees hanging from $p^*_L$
\big(except $T(v_L)$\big) and the {\em eligible} subtrees hanging from $path(v_L,v_H)$.
This edge is stored as $(x_d,y_d)$, where $y_d\in path(r_c,r')$, 
is the edge to be used in Scenario 3.
The corresponding subtree to which $x_d$ belongs is $\tau_d$.
Next, we find the edge $(x_p,y_p)$, with $x_p\in T(v_L)$ and $y_p\in path(y_d,r')$, 
which has the lowest $LCA(x_p,v_H)$ on $path(v_L,v_H)$. 
However, if $(x_d,y_d)$ does not exist, $(x_p,y_p)$ is found by considering whole of $path(r_c,r')$
instead of $path(y_d,r')$. 
The existence of back edge $(x_1,y_1)$ implies the following properties of the computed edges 
$(x_p,y_p)$ and $(x_d,y_d)$. 

\begin{lemma}
	The edge $(x_p,y_p)$, which is a back edge, always exists and when used for $p$ traversal satisfies
	\Aa~ and \Ab. 
	\label{lem:pEdgeExist}
\end{lemma}
\begin{proof}
	Recall that $(x_1,y_1)$ is the highest edge on $path(r_c,r')$ from the 
	{\em eligible} subtrees hanging from the $p^*_L$ , whereas $(x_d,y_d)$ is the 
	highest edge on $path(r_c,r')$ from a more restricted set of subtrees.
	Thus, $y_1$ is at least as high as $y_d$ on $path(r_c,r')$.
	Now, since $(x_1,y_1)$ is a back edge with $x_1\in T(v_L)$ (see remark in Scenario 1), 
	$(x_1,y_1)$ is also a valid edge for $(x_p,y_p)$ ensuring its existence of $p$ edge.
	Further, $(x_p,y_p)$ is also a back edge because $(x_1,y_1)$ is a back edge and 
	$LCA(x_p,v_H)$ is at most as high as $LCA(x_1,v_H)$ ensuring $x_p\neq v_L$. 
	
	Now, consider the $p$ traversal using $(x_p,y_p)$, which produces an untraversed path, 
	say $p'=path(par(y_p),r')$. To prove that this traversal produces only components of 
	type $C1$ and $C2$ (the condition \Aa), 
	we only need to prove 
	that any {\em eligible} subtree (subtree hanging from $p^*_P$ with an edge to $p_c$) is not connected to $p'$. 
	This is because $p'$ itself is not connected directly to $p_c$ by an edge (as $x_1\notin p_c$). 
	We first prove this property for the subtrees queried for finding $(x_d,y_d)$. 
	Since $y_p$ is at least as high as $y_d$ on $path(r_c,r')$, any such subtree will not be connected to $p'$.
	Now, we are left to prove this property for the remaining subtrees of $T(v_L)$ hanging from $p^*_P$ 
	with an edge to $p_c$. The only such subtree is $T(v_P)$, 
	the subtree hanging from $p^*_P$ which contains $T(v_H)$ (satisfying \Ab). 
	Since, among all the edges $(x,y)$ from $T(v_L)$ to $path(y_d,r')$,
	$x_p$ is the vertex with lowest $LCA(x,v_H)$, the subtree $T(v_P)$ is not connected to $p'$.
	This is because for any such edge $(x,y)$, where $x\in T(v_P)$, would have $LCA(x,v_H)$ lower than
	$v_P$, which is clearly lower than $LCA(x_p,v_H)$ on $path(v_L,v_H)$. 
\end{proof}

\begin{lemma}
	On performing $r$ traversal using any edge $(x_r,y_r)$, which satisfies (i) $x_r\in T(v_P)\cup \tau_p$ (if any), 
	and (ii) the conditions \Aa~ and \Ab, 
	the $r$ traversal is applicable with the  
	lowest edge from an {\em eligible} subtree to $p^*_R$ being $(x_d,y_d)$ (from the {\em eligible} subtree $\tau_d$), 
	if either $\tau_d\neq \tau_p$ or $\tau_d$ is not traversed by $p^*_R$. 
	\label{lem:dEdgeExist}
\end{lemma}
\begin{proof}
	Since $(x_p,y_p)$ is the edge to $path(y_d,r')$ with the lowest $LCA(x_p,v_H)$, 
	the subtree $T(v_P)$ does not have an edge to $path(y_d,r')$. 
	Thus, any subtree of $T(v_P)$, even if {\em eligible}, cannot have an edge on $path(r_c,r')$ 
	higher than $(x_d,y_d)$. As a result, if $r$ traversal is performed with $x_r\in T(v_P)\cup \tau_p$, 
	\big(see Figure~\ref{fig:pdfs} (c)\big), 
	the highest edge on $path(r_c,r')$ (and hence lowest on $p^*_R$) can be found by 
	querying the same {\em eligible} subtrees as the ones queried while computing $(x_d,y_d)$. 
	Also, if some part of $path(r_c,r')$ is not traversed by $r$ traversal, 
	it would effect this lowest edge only if it has an edge to $p_c$ or some eligible subtree,
	which is avoided by second condition. 
	Finally, it also requires the subtree containing $x_d$ to remain 
	connected to $p_c$ and leave $x_d$ untraversed.
	Since $LCA(x_p,v_H)$ is atleast as low as $LCA(x_d,v_H)$, $T(v_P)$ and $\tau_d$ are disjoint.
	Hence, the only way in which $r$ traversal traverses $\tau_d$ is if $x_p\in \tau_p$ and $\tau_p=\tau_d$.
%
\end{proof}



Thus, Lemma \ref{lem:pEdgeExist} ensures that our traversal can follow $p^*_P$ as
shown in Figure \ref{fig:pdfs} (b). To verify the third condition of {\em applicability} lemma, 
we find the new root for the component having path $p_c$ as follows. 
We find the lowest edge on $p^*_P$ from $p_c$ and the {\em eligible} subtrees hanging from $p^*_P$, 
say $(x_2,y_2)$, where $y_2\in p^*_P$. In case this edge satisfies \Ac, 
we perform the traversal otherwise move to the next scenario.

\begin{remark}
This scenario is not applicable only if $(x_2,y_2)$ is a back edge with $x_2\in T(v_P)$ and 
	$x_2\notin T(v_H)$.
\end{remark}

\subsubsection*{Scenario 3: {\large $r$} traversal}
Consider the traversal shown in Figure~\ref{fig:pdfs}~(c), where 
$p^*_R=path(r_c,x_r)\cup (x_r,y_r)\cup path(y_r,r')$. 
We choose $(x_2,y_2)$ as $(x_r,y_r)$. 
However, while computing $(x_2,y_2)$, $\tau_p$ (if exists) would have been partially traversed.
Hence, if the lowest edge from $\tau_p$ to $path(r_c,r')$, say $(x'_2,y'_2)$, has $y'_2$ lower 
than $y_r$ on $path(r_c,r')$, $\tau_p$ would be connected to both $p_c$ and 
$path(par(v_l),y_2)\backslash\{y_2\}$.
This creates a component having $p_c$ which is not of type $C_1$ or $C_2$ violating \Aa.
In such a case we choose $(x'_2,y'_2)$ as $(x_r,y_r)$.
The existence of back edge $(x_2,y_2)$ implies the following property of $(x_r,y_r)$. 


\begin{lemma}
	The edge $(x_r,y_r)$, which is a back edge, always exists and when used for $r$ traversal 
	satisfies \Aa~ and \Ab. 
	\label{lem:rEdgeExist}
\end{lemma}
\begin{proof}
	Existence of $(x_2,y_2)$ clearly implies the existence of $(x_r,y_r)$.
	Further, $(x_r,y_r)$ is a back edge since both choices for it are back edges, i.e., 
	$(x_2,y_2)$ (see remark in Scenario 2) and $(x'_2,y'_2)$ (as $root(\tau_p)\neq v_L$).  

	Now, consider the $r$ traversal using $(x_r,y_r)$, which produces an untraversed path, 
	say $p'=path(par(v_l),y_r)\backslash \{y_r\}$. 
	To prove that this traversal produces only components of type $C1$ and $C2$ 
	(hence satisfies \Aa),
 	we only need to prove that any {\em eligible} subtree (subtree hanging from $p^*_R$ with an edge to $p_c$) 
 	is not connected to $p'$. This is because $p'$ itself is not connected directly to $p_c$ by an edge 
 	(as $x_2\notin p_c$). 
	Now, the only such subtrees not queried while computing $(x_2,y_2)$ is $\tau_p$, which is queried 
	while computing $(x'_2,y'_2)$. Hence for either choice for $(x_r,y_r)$ ($(x_2,y_2)$ or $(x'_2,y'_2)$),
	no {\em eligible} subtree will be connected to $p'$ (ensuring \Ab). 
\end{proof}

Thus, Lemma \ref{lem:rEdgeExist} ensures that our traversal can follow $p^*_R$ as shown in 
Figure \ref{fig:pdfs} (c). To verify the third condition of {\em applicability} lemma, 
we find the new root for the component having path $p_c$ as follows. 
We find the lowest edge on $p^*_R$ from $p_c$ and the {\em eligible} subtrees 
hanging from $p^*_R$, say $(x_3,y_3)$, where $y_3\in p^*_R$. In case this edge satisfies 
\Ac, 
we perform this traversal.

Using Lemma~\ref{lem:dEdgeExist} we shall describe the conditions when this edge satisfies
{\em applicability} lemma. The choice of $(x_r,y_r)$ ensures that either $x_r\in T(v_P)$ 
(see Remark in $p$ traversal) or $x_r\in \tau_p$ (recall the computation of $(x'_2,y'_2)$).
Further, using Lemma~\ref{lem:rEdgeExist} our $r$ traversal satisfies \Aa~ and \Ab.
the first two conditions of {\em applicability} lemma. 
Thus, using Lemma~\ref{lem:dEdgeExist} the scenario is not applicable
only in the {\em special case} where $\tau_d=\tau_p$ and $p^*_R$ traverses $\tau_d$. 
Refer to Procedure~\ref{alg:dfs_c2} in Appendix~\ref{appn:pseudocode} for the pseudo code of this traversal.

We now present the conditions of the {\em special case} and present an overview of how it can be handled. 
Since $\tau_d=\tau_p$, $p^*_R$ would traverse $\tau_d$ only if $x_r\in \tau_d=\tau_p$, i.e., 
$(x_r,y_r)=(x'_2,y'_2)$. Thus, both the lowest and the highest edges on $path(r_c,r')$, i.e., 
$(x_p,y_p)$ and $(x_r,y_r)$, from an eligible subtree hanging from $path(v_L,v_H)$ belong to $\tau_d$.
Moreover, since  $\tau_d$ hangs from  $path(v_L,v_H)$, it does not contain $T(v_H)$. 
This ensures that if modified $r'$ traversal is performed ignoring $\tau_d$, 
it can be followed by a disconnecting traversal of $\tau_d$
described as follows.

\subsection*{Special case of heavy subtree traversal}
\label{sec:SpecialCase}
In this section $(x_p,y_p)$ and $(x_r,y_r)$ correspond to the back edges used in 
$p$ and $r$ traversals described earlier,
whereas $(x_{r'},y_{r'})$ corresponds to the back edge used in modified $r'$ traversal described below. 
We now recall the conditions leading to the {\em special} case, 
and describe its implications on $(x_p,y_p)$, $(x_r,y_r)$ and $(x_d,y_d)$.

\begin{lemma}
	The conditions of the special case are (i) $x_3\in T(v_R)$, (ii) $x_p\in \tau_d$, and 
	(iii) $(x_r,y_r)=(x'_2,y'_2)$. Following are the properties of $(x_p,y_p)$ and $(x_d,y_d)$ 
	in this case.
	\begin{enumerate}
\item $y_d=y_p$,$y_3$ is lower than $y_d$ on $path(r_c,r')$, and $y_r$ is lower than $y_2$ on $path(r_c,r')$.
\item No subtree of $\tau_d$ hanging from $path(root(\tau_d),x_p)$, with an edge to $p_c$, 
		has an edge lower than $y_2$ on $path(r_c,r')$. 
\item No subtree of $\tau_d$ hanging from $path(root(\tau_d),x_r)$, with an edge to $p_c$,
		has an edge higher than $y_3$ on $path(r_c,r')$. 	
	\end{enumerate}			
	\label{lem:propSC}
\end{lemma}
\begin{proof}
Recall the choice of $(x_d,y_d)$ and $(x_p,y_p)$ (see Section~\ref{sec:choose_ep_ed}), 
it was the highest edge from 
$\tau_d$ on $path(r_c,r')$ and $(x_p,y_p)$ was computed such that $y_p\in path(y_d,r')$. 
Hence, if $x_p\in \tau_d$ we necessarily have $y_d=y_p$.
Also, $y_3$ is strictly lower than $y_p$ (or $y_d$) else having a lower $LCA(x_3,v_H)$ 
than $LCA(x_p,v_H)$, $(x_3,y_3)$ would have been selected as $(x_p,y_p)$ earlier.
Finally, $y_2$ is strictly higher than $y_r$ otherwise $(x_2,y_2)$ 
would have been selected as $(x_r,y_r)$ earlier.
Second property holds since $(x_2,y_2)$ (where $x_2\in T(v_P)$) was the lowest edge on $path(r_c,r')$
from the eligible subtrees after $p$ traversal. 
Third property holds since $(x_3,y_3)$ (where $x_3\in T(v_R)$) was the highest edge on $path(r_c,r')$
from the eligible subtrees after $r$ traversal.
\end{proof}

\begin{figure*}[!ht]
	\centering
	\includegraphics[width=\linewidth]{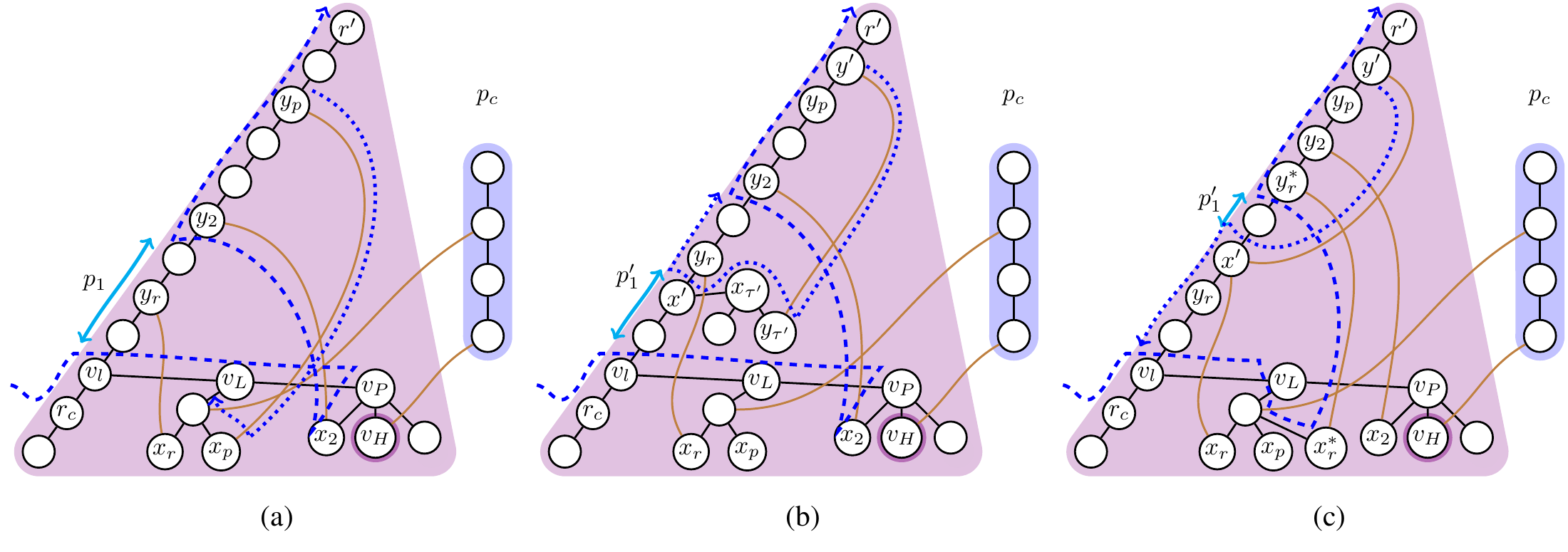}
	\caption{The three traversals for {\em special case} of Heavy Subtree Traversal
		(shown  using blue dotted lines) followed by a modified $r'$ traversal (shown  using blue dashed lines).
		(a) Root traversal of $\tau_d$, 
		(b) Upward cover traversal of $p_1$ through $\tau'$, and 
		(c) Downward cover traversal of $p_1$ using a direct edge $(x',y')$.		
	}
	\label{fig:pdfsSp}
\end{figure*}

To handle this special case, we revisit the scenario corresponding to $r$ traversal 
ignoring the {\em eligible} subtree $\tau_d$. 
Hence, we choose $(x_{r'},y_{r'})=(x_2,y_2)$ despite having a lower edge $(x_r,y_r)$.
Now, based on the lowest edge from component containing $p_c$ on the traversed path, 
we append simple traversals to this modified $r'$ traversal in order to satisfy \Aa. 
We shall shortly see that in these traversals the conditions \Ab~ and \Ac~
are implicitly true.

\subsection*{Modified {\large $r'$} traversal}
Consider a modified $r'$ traversal using $(x_{r'},y_{r'})=(x_2,y_2)$ 
traversing the path shown in Figure~\ref{fig:pdfsSp} (a), 
where $p^*_{R'}=path(r_c,x_2)\cup (x_2,y_2)\cup path(y_2,r')$. 
This leaves an untraversed part of $path(r_c,r')$, i.e., 
$path(par(v_l),y_2)\backslash\{y_2\}$ (say $p_1$).
Using Lemma~\ref{lem:propSC} we know that $y_2$ is strictly higher than $y_r$ ensuring $y_r\in p_1$.
Now, the computation of $(x_2,y_2)$ ensures that no {\em eligible} subtree (except for $\tau_d$)
is connected to both paths $p_1$, as well as ensuring the lowest edge on $p^*_{R'}$ from an 
{\em eligible} subtree is from $\tau_d$, i.e., $(x_d,y_d)$. 
This implicitly satisfies \Ab~ and \Ac,
if the appended traversal does not traverse any part of $T(v_{R'})$ or $p_c$. 
However, since $\tau_d$ is connected to two paths $p_1$ and $p_c$ 
(violating \Aa), 
we need to append $p^*_{R'}$ with another traversal which disconnects the 
unvisited part of $\tau_d$ from the unvisited part of $p_1$.


Now, the only {\em eligible} subtree having an edge to $p_1$ is $\tau_d$.
Still the lowest edge on $p^*_{R'}$ from the component having $p_c$ may not be from $\tau_d$.
This is because $p_1$ and subtrees connected to it are also in the component containing $p_c$ 
(since $\tau_d$ has an edge on $p_1$). 
Thus, depending on the lowest edge on $p^*_{R'}$ from the  component containing $p_c$, 
say $(x,y)$ with $y\in p^*_{R'}$, we have three cases (see Figure~\ref{fig:pdfsSp}).
If $x\in \tau_d$, we simply perform a {\em root traversal} of $\tau_d$ 
exploiting property 2 of Lemma~\ref{lem:propSC} to disconnect $\tau_d$ and $p_1$. 
On the other hand if $x\in p_1$ or some subtree connected to $p_1$,
we shall perform a {\em cover traversal} of $p_1$, which visits all vertices on $p_1$ that are connected to $\tau_d$.
As a result, the unvisited part of $p_1$ is disconnected from $\tau_d$.
Since $y_r$ is the lowest edge from $\tau_d$ on $p_1$, in case $x$ is lower than $y_r$ on $p_1$,
we simply traverse upwards covering $y_d$ and other endpoints of edges from $\tau_d$ incident on $p_1$.
Else we traverse downwards after making sure that no endpoints of edges from $\tau_d$ that incident on $p_1$
are above $x$.

\subsubsection*{Root Traversal of {\large $\tau_d$}}
In this case the lowest edge from component containing $p_c$ on $p^*_{R'}$ is from $\tau_d$, i.e., $(x_d,y_d)$.
The traversal of $p^*_{R'}$ is followed by the traversal of $p_{R'}=(y_p,x_p)\cup (x_p,root(\tau_d))$
as shown in Figure~\ref{fig:pdfsSp} (a).
Hence, using property~2 in Lemma~\ref{lem:propSC}, no subtree of $\tau_d$ hanging from $p_{R'}$ 
is connected to $p_c$
(satisfying  \Aa). 
As described earlier, since $T(v_{R'})$ does not have an edge to $\tau_d$, 
after the traversal of $p^*_{R'}\cup p_{R'}$, the component having $p_c$ would have 
the new root in $\tau_d$ or $p_c$ (satisfying \Ac). 

\subsubsection*{Cover traversal of {\large $p_1$}}
In this case the lowest edge from component containing $p_c$ on $p^*_{R'}$ is from $p_1$ (say $(x',y')$) 
or some subtree $\tau'$ (say $(y_{\tau'},y')$) which is connected to $p_1$, where $y'\in p^*_{R'}$.
If connected through $\tau'$, we choose the highest edge $(x',x_{\tau'})$ from $\tau'$ on $p_1$, with $x'\in p_1$.
If $x'$ is lower than $y_r$, we perform the {\em upward traversal} towards $y_2$.
Other wise we perform the {\em downward traversal} towards $v_l$.
In case of upward traversal, when connected through $tau'$, we update $(x',x_{\tau'})$ to be the lowest edge 
from $\tau'$ on $p_1$, which still maintains $x'$ to be lower than $y_r$.
Note that this choice of $(x',x_{tau'})$ ensures that $path(y_{tau'},x_{tau'})$ is a {\em disconnecting} 
traversal of $\tau'$ from $p_1$ in both upward and downward traversals.
We also define this path from $y'$ to $x'$ as $p_{\tau'}$, i.e.,
when connected through $\tau'$  (see Figure~\ref{fig:pdfsSp} (b)), we have 
$p_{\tau'}=(y',y_{\tau'})\cup path(y_{\tau'},x_{\tau'})\cup (x_{\tau'},x')$.
Otherwise, in case of direct edge (see Figure~\ref{fig:pdfsSp} (c)), we have $p_{\tau'}=(y',x')$.

%


\begin{enumerate}
\item \textbf{Upward traversal on $p_1$}\\
In case $x'$ is lower than $y_r$, the traversal of $p^*_{R'}$ is followed by the traversal of 
$p_{R'}=p_{\tau'}\cup path(x',y_2)\backslash\{y_2\}$ as shown in Figure~\ref{fig:pdfsSp}~(b). 
Since $p_{R'}$ is a {\em disconnecting} traversal of $\tau'$ from $p_1$,
the unvisited part of $p_1$, say $p'_1$, is not connected to the unvisited part of $\tau'$.
Also, $p'_1$ is not connected to $\tau_d$ and hence the component containing $p_c$ as $y_r\notin p'_1$.
Since the unvisited part of $\tau'$ is also not connected to $p_c$, \Aa~ 
is satisfied.
As described earlier, the component having $p_c$ will have the new root in $\tau_d$, 
being the only part of the component connected to $p_{R'}$. 
Since $\tau_d$ is not a heavy subtree, \Ac~ 
is also satisfied. 

\item \textbf{Downward traversal on $p_1$}\\
In case $x'$ is higher than $y_r$ and we follow the traversal downwards, 
$path(x',y_2)\backslash\{y_2\}$ might still have edges from $\tau_d$. 
Hence, we modify the traversal of $p^*_{R'}$ as follows. Let the lowest 
edge on $path(x',y_2)\setminus\{x',y_2\}$ from $\tau_d$ be $(x^*_r,y^*_r)$, 
where $y^*_r\in p_1$. In case $(x^*_r,y^*_r)$ doesn't exist, we simply choose 
$(x^*_r,y^*_r)=(x_2,y_2)$. We now perform a modified $r''$ traversal using 
$(x_{r''},y_{r''})=(x^*_r,y^*_r)$ traversing the path 
$p^*_{R''}=path(r_c,x^*_r)\cup (x^*_r,y^*_r)\cup path(y^*_r,r')$  
as shown in Figure~\ref{fig:pdfsSp} (c). Since $y^*_r$ is higher than $x'$, 
again the path from lowest edge on $p^*_{R''}$ from the component containing $p_c$, 
to $p_1$ would correspond to $p_{\tau'}$. This traversal is then followed by the 
traversal of $p_{R''}=p_{\tau'}\cup path(x',par(v_l))$ as shown in Figure~\ref{fig:pdfsSp}~(c). 
Since $p_{R''}$ is a {\em disconnecting} traversal of $\tau'$ from $p_1$, 
the unvisited part of $p_1$, say $p'_1$, is not connected to the unvisited part of $\tau'$.
Also, $p'_1$ would not be connected to $\tau_d$ and hence component containing $p_c$,
because $y^*_r$ was the lowest edge above $x'$ on $p_1$ from $\tau_d$.
Since the unvisited part of $\tau'$ is not connected to $p_c$,
\Aa~ 
is satisfied.
As described earlier, the component having $p_c$ will have the new root in $\tau_d$, 
being the only part of the component connected to $p_{R'}$. 
Since $\tau_d$ is not a heavy subtree, \Ac~ 
is also satisfied. 
\end{enumerate}

Thus, in all the cases of Special Case of heavy path traversal, one of the traversals described above 
is necessarily applicable. Refer to Procedure~\ref{alg:dfs_c2sp} in Appendix~\ref{appn:pseudocode} for 
pseudo code.

\subsection*{Correctness:}
To prove the correctness of our algorithm, it is sufficient to prove two properties.
\textit{Firstly}, the components property is satisfied in each traversal mentioned above.
\textit{Secondly},  every component in a phase/stage, 
abides by the size constraints defining the phase/stage.
By construction, we always choose the lowest edge from a component to the 
recently added path in $T^*$ ensuring that the components property is satisfied.
Furthermore, in the different traversals we have clearly proved how the 
stage/phase is progressed ensuring the size constraints.
Thus, the final tree $T^*$ returned by the algorithm is indeed the 
DFS tree of the updated graph.

\subsection*{Analysis}
We now analyze a stage of the algorithm for processing a component $c$.
In each stage, our algorithm performs at most $O(1)$ traversals of each type described above. 
Let us first consider the queries performed on the data structure $\cal D$.
Every traversal described above performs $O(1)$ sets of these queries sequentially, 
where each set may have $O(|c|)$ parallel queries (refer to Appendix~\ref{appn:pseudocode} for the pseudo code). 
Moreover, each of these sets is an {\em independent} set of parallel queries on $\cal D$
(recall the definition of {\em independent} queries in Section~\ref{sec:prelim}).
This is because in each set of parallel queries, different queries are performed 
 either on different untraversed subtrees of currently processed subtree
or on the traversed path in the currently processed subtree.
The remaining operations (excluding queries to $\cal D$) clearly requires only the knowledge of the current DFS tree $T$ (and not whole $G$). 
Hence, they can be performed locally in the distributed and semi-streaming environment. Performing 
these operations efficiently in parallel shall be described in Section~\ref{sec:p_imp}. 
Since our algorithm requires $\log n$ phases each having $\log n$ stages, we get the following theorem.

\begin{theorem}
	Given an undirected graph and its DFS tree $T$, any subtree $\tau$ of $T$ can be rerooted at any vertex $r'\in \tau$ 
	by sequentially performing $O(\log^2 n)$ sets of $O(|\tau|)$ independent queries on $\cal D$, 
	in addition to local computation requiring only the subtree $\tau$.
	\label{thm:rerootG}
\end{theorem}

\section{Implementation in the Parallel Environment}
\label{sec:p_imp}
We assign $|c|$ processors to process a component $c$, requiring overall $n$ processors. 
We first present efficient implementation of $\cal D$ and the operations on $T$ used by our algorithm.

\subsection{Basic Data Structures}
\label{sec:ds}
The data structure maintained by our algorithm uses the following classical results for 
finding the properties of a tree on an EREW PRAM.

\begin{theorem}[Tarjan and Vishkin~\cite{TarjanV84}]
A rooted tree on $n$ vertices can be processed in $O(\log n)$ time using $n$ processors to compute
post order numbering of the tree, level and number of descendants for each vertex on a EREW PRAM.
\label{thm:eulerT}
\end{theorem}

\begin{theorem}[Schieber and Vishkin~\cite{SchieberV88}]
A rooted tree on $n$ vertices can be preprocessed in $O(\log n)$ time using $n$ processors on an EREW PRAM
such that $k$ LCA queries can be answered in $O(1)$ time using $k$ processors on a CREW PRAM.
\label{thm:lca}
\end{theorem}

Using the standard simulation model~\cite{JaJa92} for converting a CRCW PRAM algorithm to EREW PRAM algorithm
at the expense of extra $O(\log n)$ factor in the time complexity, we get the following theorem.

\begin{theorem}
A rooted tree on $n$ vertices can be preprocessed in $O(\log n)$ time using $n$ processors on an EREW PRAM
such that any $k$ LCA queries can be answered in $O(\log n)$ time using $k$ processors on an EREW PRAM.
\label{thm:lca-adv}
\end{theorem}

%

We also use the following classical result to sort and hence report maximum/minimum of a set of $n$ numbers
on an EREW PRAM. 

\begin{theorem}[Cole~\cite{Cole88,Cole93}]
A set of $n$ numbers can be sorted using parallel merge sort in $O(\log n)$ time using $n$ processors on an EREW PRAM.
\label{thm:sort}
\end{theorem}

\subsection{Implementation of $\cal D$}
\label{sec:dsD}
Given the DFS tree $T$ of the graph, we build the data structures 
described in Theorem~\ref{thm:eulerT} and Theorem~\ref{thm:lca-adv} on it. 
Now, given the post order traversal of $T$, 
we assign each vertex with a value equal to its rank in the post order traversal. 
Now, for each vertex $v$ we perform a parallel merge sort on the set of 
neighbors of the vertex using $degree(v)$ processors, requiring overall $m$ processors.
Thus, each vertex stores its neighbors \big(say $N(v)$\big) in the 
increasing order of their post order indexes. Due to absence of cross edges in a DFS tree $T$, 
the neighbors of every vertex would be sorted in the order they appear on 
the path from $root(T)$ to the vertex. Thus, the data structure $\cal D$ can be built 
in $O(\log n)$ time (for sorting) using $m$ processors on an EREW PRAM. 
This allows us to answer the following queries efficiently.

\begin{enumerate}
\item $Query\big(w,path(x,y)\big):$ among all the edges from $w$ that are incident on $path(x,y)$ 
					 in $G$, return an edge that is incident nearest to $x$ on $path(x,y)$.%
\item $Query\big(T(w),path(x,y)\big):$ among all the edges from $T(w)$ that are incident on $path(x,y)$ 
					in $G$, return an edge that is incident nearest to $x$ on $path(x,y)$.
\item $Query\big(path(v,w),path(x,y)\big):$ among all the edges from $path(v,w)$ that are incident on 
					 $path(x,y)$ in $G$, return an edge that is incident nearest to $x$ on $path(x,y)$.%
\end{enumerate}

We now describe how to perform a set of {\em independent} queries to $\cal D$ (recall definition of {\em independent} queries in Section~\ref{sec:prelim}) 
in $O(\log n)$ time on an EREW PRAM as follows.
We assign one processor for each vertex $u\in \{w\},T(w)$ or $path(x,y)$ 
(depending on the type of query) to perform the following in parallel.
For the vertex $u$, we would first perform a binary search for the range 
given by the post order indexes of $x$ and $y$ on $N(u)$ to find the required edge. 
However, since all vertices of $path(x,y)$ may not be ancestors of $u$, 
$N(u)$ may include some edges not on $path(x,y)$ too in the given range, 
corrupting the search results. Hence, the search would be performed on a 
modified range described as follows. Firstly, assuming $x$ is an ancestor of $y$, 
if $LCA(u,x)$ is not equal to $x$ the search would not be performed (as $x$ is not an ancestor of $u$).
Otherwise, the search is performed on the range given by post order indexes of $x$ and $LCA(u,y)$. 
However, in case of $Query\big(path(v,w),path(x,y)\big)$ we surely know 
that no vertex of $path(v,w)$ is a descendant of $path(x,y)$ (recall its definition in Section~\ref{sec:prelim}).
Thus, we reverse the roles of the paths taking maximum or minimum accordingly using $|path(x,y)|$ processors. 
Thus, each of these queries would require $O(\log n)$ time on an EREW PRAM.
Now, given a set of \emph{independent} queries on $\cal D$, each processor 
shall be using different $N(u)$ for finding the corresponding edge. 
Hence, all the queries can be performed simultaneously on different memory cells 
abiding the constraints of an EREW PRAM. 
Now, the highest or lowest edge 
among all the edges returned by different processors can be found by taking the 
maximum or minimum in $O(\log n)$ time on an EREW PRAM (Theorem~\ref{thm:sort}). 
Thus, we have the following theorem.

\begin{theorem}
The DFS tree $T$ of a graph can be preprocessed to build a data structure $\cal D$ 
of size $O(m)$ in $O(\log n)$ time using $m$ processors such that a set of independent queries of
types $Query\big(w,path(x,y)\big)$, $Query\big(T(w),path(x,y)\big)$ and $Query\big(path(v,w),path(x,y)\big)$  
on $T$ can be answered simultaneously in $O(\log n)$ time using $1,|T(w)|$ and $|path(x,y)|$ processors 
respectively on an EREW PRAM.
\label{thm:DS}
\end{theorem}

\subsubsection*{Extension to handle multiple updates}
Consider a sequence of $k$ updates on graph, let $T^*_i$ represent the DFS tree computed by our algorithm 
after $i$ updates in the graph. We also denote the corresponding data structure $\cal D$ built on $T^*_i$ 
as ${\cal D}_i$. We now show that any query of the type $Query\big(w,path(x,y)\big)$, 
$Query\big(T^*_i(w),path(x,y)\big)$ and $Query\big(path(v,w),path(x,y)\big)$ on ${\cal D}_i$, 
can be performed on ${\cal D}_0$ if $path(x,y)$ is an ancestor-descendant path in $T$.
Recall that each such query is performed by querying the $N(x)$ corresponding to each 
descendant vertex $x$ separately, whose results are later combined. 
Thus, even if $T^*_i(w)$ is not a subtree of $T$ or $path(v,w)$ is not an ancestor-descendant path of $T$,
it does not affect the processing of the query, as long as $path(x,y)$ is an ancestor-descendant path of $T$. 

The only extra procedure to be performed to answer such queries correctly using ${\cal D}_0$, 
is to update the $N(x)$ for any vertex $x$ whose adjacency list is affected by the graph update. 
For insertion/deletion of a vertex $x$, we simply add/delete the corresponding list $N(x)$.
For insertion of vertex we additionally sort it according to post order traversal of $T$ 
using $n$ processors in $O(\log n)$ time. 
Note that we do not need to update the $N(y)$ for each neighbor $y$ of $x$, 
as the query path being an ancestor-descendant path of both $T^*_i$ and $T$ would not contain $x$.
However, on insertion of a vertex $x$, such a query can be made with the entire path representing only $x$.
Hence, we assign the highest post order number to $x$, and add it to the end of $N(y)$ for each neighbor $y$ of $x$.
This can be done using $n$ processors in $O(1)$ time on an EREW PRAM. 
Insertion/deletion of single edges can be taken care of individually by each search 
procedure taking $O(\log n+k)$ time to perform search after $k$ updates. Thus, we have the following theorem

\begin{theorem}
The data structure $\cal D$ built on the DFS tree $T$ of a graph $G$, 
can be used to perform a set of independent queries on ${\cal D}_k$ of types 
$Query\big(w,path(x,y)\big)$, $Query\big(T^*_k(w),path(x,y)\big)$ and 
$Query\big(path(v,w),path(x,y)\big)$, in $O(\log n+k)$ 
time using $1,|T^*_k(w)|$ and $|path(x,y)|$ processors respectively on an EREW PRAM,
if $path(x,y)$ is an ancestor-descendant path of $T$.
\end{theorem}

\subsection{Implementation of operations on $T$}
As described earlier several properties of $T$ can be reported in $O(1)$ time 
using the data structures described in Theorem~\ref{thm:eulerT} and Theorem~\ref{thm:lca-adv}.

\begin{enumerate}
	\item \textbf{Determine whether an edge $(x,y)$ is a back edge in $T$}\\
	This query can easily be answered by finding $l= LCA(x,y)$. 
	If $l=x$ or $l=y$ the edge $(x,y)$ is a back edge in $T$.
	Hence, reporting whether an edge is a back edge can be reduced to finding $LCA$
	 of two vertices in $T$.
	\item \textbf{Finding length of a path}\\
	Compare the level of the two end points as reported by structure in Theorem~\ref{thm:eulerT}.
	\item \textbf{Given $x\in T(y)$, find child $y'$ of $y$ such that $x\in T(y')$}\\
	For each vertex $v$ of the graph perform the following in parallel (using $|T(y)|$ processors), 
	if $par(v)$ is $y$ and $LCA(v,x)$ is $v$ then report $v$. This query too reduces to 
	finding $LCA$ of two vertices in $T$.
	\item \textbf{Determine whether $x$ lies on $path(y,z)$, where $y$ is ancestor of $z$}\\
	If $LCA(x,z)=x$ and $LCA(x,y)=y$, then $x$ lies on $path(y,z)$.
	\item \textbf{Subtrees hanging from a $path(x,y)$}\\
	For each vertex $v$ of the graph perform the following in parallel (using total $n$ processors),
	if $LCA(v,y)=par(v)$ then $T(v)$ is a subtree hanging from the path.
\end{enumerate}

The number of processors required for the last three queries is equal to the size of the corresponding component,
remaining queries requiring a single processor each.
Thus, using Theorem~\ref{thm:sort} and procedures described above we have the following theorem

\begin{theorem}
The DFS tree $T$ of a graph can be preprocessed to build a data structure of size $O(n)$ in $O(\log n)$ time 
using $n$ processors such that following queries can be answered in parallel in $O(\log n)$ time on an EREW PRAM
\begin{itemize}
\item LCA of two vertices, size of a subtree, testing if an edge is back edge and length of a path
using a single processor per query.
\item Finding vertices on a path, subtrees hanging from a path, child of a vertex containing a given vertex,  
highest/lowest edge among $k$ edges, using $k$ processors per query, where k is the size of the corresponding component.
\end{itemize}
\label{thm:DSx}
\end{theorem}

\subsection{Analysis}
Using these data structures we can now analyze the time required by the 
{\em reduction} algorithm on an EREW PRAM.
Since the queries on $\cal D$ and $LCA$ queries 
on $T$ can be answered  in $O(\log n)$ time using $n$ processors 
as described above,
Theorem~\ref{thm:convertApp} reduces to the following theorem.


\begin{theorem}
Given the DFS tree $T$ of a graph and the data structure $\cal D$ built on it, 
any update on the graph can be reduced to independently rerooting disjoint subtrees 
of the DFS tree using $n$ processors in $O(\log n)$ time on an EREW PRAM.
\label{thm:convert}
\end{theorem}

\subsubsection*{Implementation details}
All operations required for each stage of our rerooting algorithm to reroot a subtree $\tau$, 
can be performed in $O(\log n)$ time using $|\tau|$ processors using Theorem~\ref{thm:DS} and Theorem~\ref{thm:DSx} as follows.
Both $root(\tau_c)$ and vertex $v_H$ required by our algorithm while processing a component $c$ can be computed in parallel 
by comparing the size of each subtree using $|c|$ processors. Adding a path $p$ to $T^*$ essentially involves
marking the corresponding edges as tree edges, which can be performed by informing the vertices on $p$.
All the other operations of the rerooting algorithm (refer to pseudo code in Appendix~\ref{appn:pseudocode}) 
are trivially reducible to the operations described in Theorem~\ref{thm:DSx}. 
Since our rerooting algorithm requires $\log n$ phases each having $\log n$ stages, we get the following theorem for 
rerooting disjoint subtrees using our rerooting algorithm.
 
\begin{theorem}
	Given an undirected graph with the data structure $\cal D$ build on its DFS tree, 
	independently rerooting disjoint subtrees of the DFS tree can be performed in $O(\log^3 n)$ 
	time using $n$ processors on an EREW PRAM model. 
	\label{thm:reroot}
\end{theorem}

Using Theorem~\ref{thm:convert}, Theorem~\ref{thm:reroot} and Theorem~\ref{thm:DS}, we can prove our main result as follows.

\newtheorem*{ThmMainResult}{Theorem~\ref{main-result}}
\begin{ThmMainResult}
Given an undirected graph and its DFS tree, it can be preprocessed to
	 build a data structure of size $O(m)$ in $O(\log n)$ time using $m$ processors on an EREW PRAM 
	 such that for any update in the graph, a DFS tree of the updated graph can be computed in 
	 $O(\log^3 n)$ time using $n$ processors on an EREW PRAM.
\end{ThmMainResult}

Now, in order to prove our result for Parallel Fully Dynamic DFS and Parallel Fault Tolerant DFS  
we need to first build the DFS tree of the original graph from scratch during preprocessing stage.
This can be done using the static DFS algorithm~\cite{Tarjan72} or any advanced deterministic parallel 
algorithm~\cite{AggarwalA88,GoldbergPV93}. 
Thus, for processing any update we always have the current DFS tree built 
(either originally during preprocessing or by the update algorithm).
We can thus build the data structure $\cal D$ using Theorem~\ref{thm:DS} 
reducing Theorem~\ref{main-result} to the following theorem.

\begin{theorem}[Parallel Fully Dynamic DFS]
	Given an undirected graph, we can maintain its DFS tree under any arbitrary online 
	sequence of vertex or edge updates in $O(\log^3 n)$ time per update using $m$ processors 
	in parallel EREW model.
\end{theorem}

However, if we limit the number of processors to $n$, our fully dynamic algorithm cannot update the DFS tree 
in $\tilde{O}(1)$ time, only because updating $\cal D$ in $\tilde{O}(1)$ time requires $O(m)$ processors 
(see Theorem~\ref{thm:DS}).
Thus, we build the data structure $\cal D$ using Theorem~\ref{thm:DS} during preprocessing itself,
and attempt to use it to handle multiple updates.

\subsection*{Extending to multiple updates}
Consider a sequence of $k$ updates on graph, let $T^*_i$ represent the DFS tree computed by our algorithm 
after $i$ updates in the graph. We also denote the corresponding data structure $\cal D$ built on $T^*_i$ as ${\cal D}_i$.
Now, consider any stage of our algorithm while building the DFS tree $T^*_i$.
For each component in parallel, $O(1)$ ancestor-descendant paths of $T^*_{i-1}$ are added to $T^*_{i}$. 
Thus, any ancestor-descendant path $p$ of $T^*_i$, is built by adding $O(\log^2 n)$ such paths of $T^*_{i-1}$,
corresponding to $O(\log n)$ phases each having $O(\log n)$ stages.
Hence, $p$ is union of $O(\log^2 n)$ ancestor-descendant paths of $T^*_{i-1}$, say $p_1,...,p_k$.

Using this reduction, it can be shown that a set of independent queries on path $p$ in ${\cal D}_i$,
can be reduced to $O(\log^2 n)$ sets of independent queries on corresponding $O(\log^2 n)$ paths $p_1,...,p_k$
on ${\cal D}_{i-1}$ (see Section~\ref{sec:dsD}).  
Again, each of these paths $p_1,...,p_k$, being an ancestor-descendant path of $T^*_{i-1}$,
is a union of $O(\log^2n)$ ancestor-descendant paths of $T^*_{i-2}$, and so on.
Thus, any set of independent queries on ${\cal D}_i$ can be performed by $O(\log^{2(i-1)} n)$ sets of 
independent queries on $\cal D$, which takes $O(\log^{2i-1})$ time on an EREW PRAM using $n$ processors when $k\leq \log n$ 
(see Theorem~\ref{thm:DS} and Section~\ref{sec:dsD}). The other data structures on $T^*_{i-1}$ can be built in $O(\log n)$ time using $n$
processors (see Theorem~\ref{thm:DSx}). This allows our algorithm to build the DFS tree $T^*_i$ from $T^*_{i-1}$ 
using $\cal D$ in $O(\log^{2i+1})$ time on an EREW PRAM using $n$ processors (see Theorem~\ref{thm:rerootG}). 
Thus, for a given set of $k$ updates we build each $T^*_i$ one by one using $T^*_{i-1}$ and $\cal D$, 
to get the following theorem.

\begin{theorem}[Parallel Fault Tolerant DFS]
	 Given an undirected graph, it can be preprocessed to
	 build a data structure of size $O(m)$ such that for any set of $k$ ($\leq \log n$) updates in the graph, 
	 a DFS tree of the updated graph can be computed in $O(k\log^{2k+1} n)$ time using 
	 $n$ processors on an EREW PRAM.
\end{theorem}

\begin{remark} 
For $k=1$, our algorithm also gives an $O(n\log^3 n)$ time sequential algorithm for updating a DFS tree 
after a single update in the graph, achieving similar bounds as Baswana et al.~\cite{BaswanaCCK15}. 
However, our algorithm uses much simpler data structure $\cal D$ at the cost of a more complex algorithm.
\end{remark}

\section{Applications in other models of computation}
We now briefly describe how our algorithm can be easily adopted to the semi-streaming model
and distributed model.

\subsection{Semi-Streaming Setting}
Our algorithm only stores the current DFS tree $T$ and the partially built DFS tree $T^*$ taking $O(n)$ space.
Thus, all operations on $T$ can be performed without any passes over the input graph.
A set of independent queries on $\cal D$ is evaluated by performing a single pass over all the edges 
of the input graph using $O(n)$ space. This is because each set has $O(n)$ queries (see Theorem~\ref{thm:convertApp} and Theorem~\ref{thm:rerootG}) and we are required to store only one edge per query (partial solution based on edges visited by the pass). 
Note that here the role of $\cal D$ is performed by a pass over the input graph. 
Hence, the algorithm is first executed for all the components in turn until each instance of the algorithm 
queries the data structure $\cal D$. This is followed by a pass on the input graph to answer these queries and so on. 
Since each stage requires $O(1)$ steps (and hence $O(1)$ sequential queries on $\cal D$), 
it can be performed using $O(1)$ passes. Thus, our algorithm requires $O(\log^2 n)$ passes 
to update the DFS tree after a graph update by executing $\log n$ stages for each of the $\log n$ phases.
Thus, we get the following theorem.

\begin{theorem}
	Given any arbitrary online sequence of vertex or edge updates, 
	we can maintain a DFS tree of an undirected graph 
	using $O(\log^2 n)$ passes over the input graph per update 
	by a semi-streaming algorithm using $O(n)$ space.
	\end{theorem}

\subsection{Distributed Setting}
\label{sec:dist-overview}
Our algorithm stores only the current DFS tree $T$ and the partially built DFS tree $T^*$ at each node.
Thus, the operations on $T$ are performed locally at each node and the distributed
computation is only used to evaluate the queries on $\cal D$. 
Using Theorem~\ref{thm:convertApp} and Theorem~\ref{thm:rerootG}, each update is performed by 
$O(\log^2 n)$ sequential sets of $O(n)$ independent queries on $\cal D$.
Evaluation of a set of $O(n)$ independent queries on $\cal D$ can be essentially reduced to 
propagation of $O(n)$ words (partial solutions of $n$ queries) throughout the network.
Using the standard technique of pipelined broadcasts and convergecasts~\cite{Peleg00}, 
we can propagate these $O(n)$ words in $O(D)$ rounds using messages of size $O(n/D)$. 
This proves our distributed algorithm described in Section~\ref{sec:results}.
We now describe the implementation details in the distributed model.

%

\subsubsection*{Implementation in distributed environment}
\label{sec:distributed}

We now present how our algorithm can be implemented on the distributed model.
In the synchronous ${\cal CONGEST}(B)$ model a processor is present at every node of the graph
and communication links are restricted to the edges of the graph.
The communication occurs in synchronous rounds, where each nodes can send a message of $O(B)$ words 
along each communication link. 
Our model includes a {\em preprocessing} stage followed by an alternating sequence of {\em update} and 
{\em recovery} stages. The graph is updated in the {\em update} stage, after which the {\em recovery} stage starts
in which the algorithm updates the DFS tree of the graph.
The model allows the algorithm to complete updating the DFS tree (completing the recovery stage)
before the next update is applied to the graph (update stage).
Similar model was earlier used by Henzinger et al.~\cite{HenzingerKN13}.
We use an additional constraint of a space restriction of $O(n)$ size at each node.
In the absence of this restriction, the whole graph can be stored at each node, 
where an algorithm can trivially propagate the update to each node and the updated solution can be computed locally.
Finally, we also allow the deletion updates to be {\em abrupt}, i.e., 
the deleted link/node becomes unavailable for use instantly after the update. 

Each node stores the current DFS tree $T$ and the partially built DFS tree $T^*$.
Thus, all the operations on $T$ can be performed locally at each node, where the distributed
computation is used only to evaluate the queries on $\cal D$. 
Also, using Theorem~\ref{thm:convertApp} and Theorem~\ref{thm:rerootG} each update reduces to $O(\log^2 n)$ 
sequential sets of $O(n)$ independent queries on $\cal D$.
Thus, we shall only focus on how to evaluate such queries efficiently in the distributed environment.


\subsubsection{Optimality of message size}
We first prove that any distributed algorithm maintaining the DFS tree at each node
requires a message size of $\Omega(n/D)$ to update the DFS tree in $O(D)$ rounds. 
Consider the insertion of a vertex, such that the final DFS tree uses $O(n)$ of the newly inserted edges.
This is clearly possible if the current DFS tree has $O(n)$ branches, where leaf of each branch is 
connected to the inserted vertex.
Thus, the information of at least these $O(n)$ new edges needs to be propagated throughout the network
by any algorithm maintaining DFS tree at each node.
Now, broadcasting $m$ messages on a network with diameter $D$ requires $\Omega(m+D)$ rounds~\cite{Peleg00}.
In order to limit the number of rounds to $O(D)$, we can send only $O(D)$ messages.
Thus, any algorithm sending $O(n)$ words of information using $O(D)$ messages would require a message size of $\Omega(n/D)$.
We thus use the ${\cal CONGEST}(n/D)$ model for our distributed algorithm.

\subsubsection{Evaluation of queries on $\cal D$}
Now, each node only stores the adjacency list of the corresponding vertex in addition to $T$ and $T^*$
described above. Recall that a query on $\cal D$ is merely highest/lowest edge among a set of eligible edges.
Hence, it can be easily evaluated for the whole graph by combining the partial solutions of the same query 
performed on each adjacency list locally at the node. Thus, the focus is to broadcast the partial solution 
from each node to reach the whole graph, where each node can then combine them locally to get the solution to the query. 
Moreover, these partial solutions can also be combined during broadcasting 
itself to avoid sending too many messages as described below.

\subsubsection*{Performing broadcasts efficiently} 
Broadcasts can be performed efficiently by using a spanning tree of the graph.
To ensure efficiency of rounds we use a BFS tree as follows.
After every update, any vertex (say vertex with the smallest index) 
starts building a BFS tree ${\cal B}$ rooted at it. The depth of $\cal B$ is $O(D)$ and 
it can be built in $O(D)$ rounds using $O(m)$ messages~\cite{Peleg00}. 
All the broadcasts are now performed only on the tree edges of $\cal B$ as follows.
We first describe it for a single query then extend it to handle $O(n)$ queries. 
Note that it is a trivial extension of the standard pipelined broadcasts and convergecasts algorithm~\cite{Peleg00}.
Each node waits for partial solutions to the query from all its children in $\cal B$,
updates its solution and sends it to its parent. On receiving the partial solutions from all the children, the root
computes the final solution and sends it back to all nodes along the tree edges of $\cal B$. 
Clearly, this process requires $O(D)$ rounds and $O(n)$ messages each of size $O(1)$ (partial solution of a query is a single edge).
In order to perform $O(n)$ independent queries efficiently in parallel, on each edge we send $D$ messages 
of size $O(n/D)$ in a pipelined manner (one after the other) to achieve the broadcast in $O(D)$ time (see pipelined broadcast in~\cite{Peleg00}). 
The total number of messages sent would be $O(nD)$. 
Since the rerooting algorithm requires $O(\log^2 n)$ sequential sets of $O(n)$ queries (see Theorem~\ref{thm:rerootG}), 
we get the following theorem.
 
\begin{theorem}
	Given any arbitrary online sequence of vertex or edge updates, 
	we can maintain a DFS tree in $O(D\log^2 n)$ rounds per update 
	in a distributed setting using $O(nD\log^2 n+m)$ messages each of size $O(n/D)$
	and $O(n)$ local space on each processor, where $D$ is diameter of the graph. 
\end{theorem}

\begin{remark} 
Our initial assumption of adding a pseudo root (see Section~\ref{sec:prelim}) connected to 
every vertex of the graph is no longer valid in the distributed system. 
This is because both the processors and communication links are fixed in our model.
Thus, we need to maintain a DFS forest instead of a DFS tree requiring to handle the cases
when some component is partitioned into several components and when two or more components merge 
as a result of a graph update. The following section describes how this can be achieved in the same bounds 
described above.
\end{remark} 

\subsubsection*{Maintaining a DFS forest} 
After every update in the graph, a neighboring vertex of the affected link/node shall broadcast the information about the update 
to all the vertices in the component. However, in order to limit the number of messages 
transmitted, exactly {\em one} vertex from each component so formed needs to initiate the broadcast.
We shall shortly describe how to choose this vertex.
The chosen vertex also chooses the  new root for the DFS tree of the component (say the node with the smallest index).
The new root then makes the corresponding BFS tree as described above to perform efficient broadcasts.
In case two or more components are merged due to the update, the DFS tree of each component computed
earlier is broadcasted to the entire component by the original roots of two components.
Since, the size of broadcast (DFS tree) is $O(n)$, it can be performed under the same bounds as described above.

We now describe how to choose the broadcast vertex efficiently.
In case of vertex/edge insertion, we choose the inserted vertex or 
endpoint of the inserted edge with smaller index respectively.
In case of vertex/edge deletion, for each component so formed, we choose the neighbor of deleted node/link in $T$ that 
has the smallest index. 
For this each neighbor of the deleted node/link needs to know the resultant components formed as a result of the deletion.
This can be easily computed locally if each node also stores the articulation points/bridges of the current DFS tree $T$.
Hence, after computing the DFS tree, each node computes the articulation points/bridges of the DFS tree
according of the subgraph induced by the edges of $T$ and the adjacency list stored at the node.
The vertices/edges present in {\em all} the sets of articulation points/bridges computed at different 
nodes will be the articulation points/bridges of the whole graph. 
Again, this requires each vertex to send $O(n)$ words of information where the partial solutions can be combined.
Thus, it can be performed similar to the queries on $\cal D$ using the same bounds.

\section{Conclusion}
Our parallel dynamic algorithms take nearly optimal time on an EREW PRAM. 
However, the work efficiency of our fully dynamic algorithm is $\tilde{O}(m)$ whereas that 
of the best sequential algorithm~\cite{BaswanaCCK16} is $\tilde{O}(\sqrt{mn})$.
Even though our fault tolerant algorithm is nearly work optimal, its only for constant number of updates.
The primary reason behind these limitations is the difficulty in updating the data structure $\cal D$
using $n$ processors. Our fault tolerant algorithm avoids updating $\cal D$, by naively using 
the original $\cal D$ to simulate the queries of updated $\cal D$. 
It would be interesting to see if an algorithm can process significantly more updates 
using only $n$ processors  in $\tilde{O}(1)$ time (similar extension was performed by 
Baswana et al.~\cite{BaswanaCCK16} in the sequential setting). This may also lead to 
a fully dynamic algorithm that is nearly time optimal with better work efficiency. 

Further, 
our distributed algorithm works only on a substantially restricted synchronous ${\cal CONGEST}(n/D)$ model. Moreover, the number of messages passed during an update in the distributed algorithm is 
$O(nD\log^2 n+m)$, which is way worse than the number of messages required to compute a DFS from 
scratch i.e. $O(n)$. It would be interesting to see if dynamic DFS can be maintained in near optimal 
rounds in more stronger ${\cal CONGEST}$ or ${\cal LOCAL}$ models.
\section*{Acknowledgement}
I am grateful to Christoph Lenzen for suggesting this problem.
I would also like to express my sincere gratitude to my advisor 
Prof. Surender Baswana for valuable discussions and key insights 
that led to this paper. The idea of extending the parallel 
algorithm to the distributed setting was rooted in these discussions.
Finally, I would like to thank the anonymous reviewers and 
my advisor whose reviews significantly helped me to refine the paper.

\bibliography{paper}

\appendix
\section{Pseudo codes of Traversals in Rerooting Algorithm}
\label{appn:pseudocode}

	\begin{procedure*}[!h]
		\tcc{Let current phase be $\PP_i$ and current stage be $\St_j$}
		$\tau_c \leftarrow$ Heaviest tree in $\TT_c$\;
		$T(v_H)\leftarrow$ Smallest subtree having size at least $n/2^i$\;
		\lIf{$|\tau_c|\leq n/2^{i}$}
		{Return Reroot-DFS($r_c,p_c,\TT_c$) in next phase}
		\lIf{$|p_c|\leq n/2^{j}$}
		{Return Reroot-DFS($r_c,p_c,\TT_c$) in next stage}
		\BlankLine
		\tcc{Disintegrating Traversal}
		\lIf{$p_c=\phi$ or $r_c= root(\tau_c)$}
		{Return DisInt-DFS($r_c,p_c,\TT_c$)}
		\BlankLine
		\tcc{Disconnecting Traversal}
		\lIf{$r_c\notin \Th_c\cup\{p_c\}$ or $r_c\in T(v_H)$}
		{Return DisCon-DFS($r_c,p_c,\TT_c$)}
		\BlankLine		
		\lIf(\tcc*[f]{Path Halving}){$r_c\in p_c$}{Return Path-Halving($r_c,p_c,\TT_c,\phi$)}
		Heavy-DFS($r_c,p_c,\TT_c$)
		\tcc*{Heavy Subtree Traversal}
		
		\caption{Reroot-DFS($r_c,p_c,\TT_c$): Traversal enters through $r_c$ 
			into the component $c$ containing a path $p_c$ and set of trees $\TT_c$.}
		\label{alg:parallel_dfs}
	\end{procedure*}

	\begin{procedure*}
		\ForEach(\tcc*[f]{$p=path(x,y)$, where $x$ lower in $T^*$})
		{$p \in {\cal P}$}
		{
			\ForEach{$\tau\in \TT$ in parallel using $|\tau|$ processors}
			{
				\If(\tcc*[f]{$\exists$ edge from $\tau$ to $p$}){$Query\big(\tau,path(x,y)\big)\neq \phi$}
				{
					$\TT\leftarrow \TT\setminus\{\tau$\}, $\TT_p\leftarrow \TT_p\cup \{\tau\}$\;
				}
			}
		
			\For(\tcc*[f]{$p'=path(x',y')$, where $x'$ lower in $T^*$})
			{$p' \in \{p_3,p_2,p_1\}$ }
			{
				$\{x_p,y_p\} \leftarrow Query\big(p,path(x',y')\big)$
				\tcc*{where $x_p\in p$}
			
				\lForEach(\tcc*[f]{where $x_{\tau}\in \tau$}){$\tau\in \TT_p$}
				{
					$(x_{\tau},y_{\tau})\leftarrow$  $Query\big(\tau,path(x',y')\big)$	
				}
				$\{x_p,y_p\} \leftarrow $ Lowest edge on $T^*$ among $(x_p,y_p)$ and $(x_\tau,y_\tau),\forall \tau\in \TT_p$\;
			
				\lIf{$(x_p,y_p)$ is a valid edge}{break}
			}
		Add $(x_p,y_p)$ to $T^*$\;
		Reroot-DFS($x_p,p,\TT_p$) in current stage\;
		}

		\ForEach{$\tau\in \TT$ in parallel using $|\tau|$ processors}
		{
			$\TT\leftarrow \TT\setminus \tau$\;
			\For(\tcc*[f]{$p'=path(x',y')$, where $x'$ lower in $T^*$})	
			{$p' \in \{p_3,p_2,p_1\}$ }	
			{
				$(x_{\tau},y_{\tau})\leftarrow$  $Query\big(\tau,path(x',y')\big)$
				\tcc*{where $x_{\tau}\in \tau$}
				\lIf{($x_\tau,y_\tau)$ is a valid edge}{break}
			}
			Add $(y_{\tau},x_{\tau})$ to $T^*$\;
			Reroot-DFS($x_{\tau},\phi,\{\tau\}$) in next stage\;	
		}
		\caption{Process-Comp(${\cal P},\TT,p^*$): 
			Moves components created of type $C1$ and components created with $p\in {\cal P}$ of type $C2$ to the next stage,
			after traversal of $p^*=p_1\cup p_2\cup p_3$, the newly added path in $T^*$.
			Here, $p\in {\cal P},p_1,p_2$ and $p_3$ are ancestor-descendant paths of $T$ and traversal of $p^*$
			ensures components of type $C1$ and $C2$ with paths in $\cal P$ only.
		}
		\label{alg:dfs_procC}
	\end{procedure*}
	\begin{procedure*}
		$T(v_H)\leftarrow$ Smallest subtree $\tau'$ of $\tau$, where $|\tau'|>n/2^i$\;
		
		$\TT\leftarrow$ Subtrees hanging from $path\big(r_c,root(\tau)\big)$\;
		$T(v_h)\leftarrow $ Subtree from $\TT$ containing $v_H$\;
		$\TT\leftarrow \TT\setminus T(v_h) \cup$ Subtrees hanging from $path(v_h,v_H)$\;
	
		Add $path(r_c,v_H)$ to $T^*$\;
		\BlankLine	
		\lIf(\tcc*[f]{Component of type $C2$}){$|p_c|\neq 0$}{$p\leftarrow p_c$;
				 $\TT\leftarrow \TT\cup\TT_c\backslash \tau$}
		\lElse(\tcc*[f]{remaining part of $path\big(r_c,root(\tau)\big)$})
				{	$p\leftarrow path\Big(par\big(par(v_h)\big),root(\tau)\Big)$}
		\BlankLine
		Process-Comp\big($\{p\},\TT,path(r_c,v_H)$\big)
		\tcc*{Goes to DisCon-DFS or next phase}
		\caption{DisInt-DFS($r_c,p_c,\TT_c$): Disintegrating Traversal of a component $c$
			having a path $p_c$ and a set of trees $\TT_c$ through the root $r_c\in \tau\in \mathbb{T}_c$, 
			where either $|p_c|=0$ or $r_c=root(\tau)$.
		}
		\label{alg:dfs_c1}
	\end{procedure*}
	\begin{procedure*}
		
		$p_c\leftarrow p_c\setminus path(r_c,x)$
		\tcc*{$p_c = path(x,y)$ where $|path(x,r_c)|\geq |path(y,r_c)|$}

		$p^*\leftarrow path(r_c,x)$\;
		
		Add $p^*$ to $T^*$\;
		Process-Comp\big($\{p_c\},\TT_c,p^*$\big)
		\tcc*{Goes to next stage}

		\caption{Path-Halving-DFS($r_c,p_c,\TT_c$): Traversal of a component $c$
			having a path $p_c$ and a set of trees $\TT_c$ through the root $r_c\in p_c$.}

		\label{alg:dfs_c2_3}
	\end{procedure*}
	\begin{procedure*}
		\BlankLine
		\tcc{$p_c=path(u,v)$, where $u$ is ancestor of $v$}
		\uIf
		{$\tau$ has an edge to upper half of $p_c$}
		{	$(x,y)\leftarrow$ Lowest edge from $\tau$ to $p_c$;
			$p'_c\leftarrow path(y,u)$
			\tcc*{where $x\in \tau$}
		}
		\lElse(\tcc*[f]{where $x\in \tau$})
		{$(x,y)\leftarrow$ Highest edge from $\tau$ to $p_c$;				
			$p'_c\leftarrow path(y,v)$
			}

		\BlankLine
	
		$\TT\leftarrow$ Subtrees hanging from $path\big(r_c,root(\tau)\big)$\;
		$T(v)\leftarrow $ Subtree from $\TT$ containing $x$\;
		$\TT\leftarrow \TT\setminus T(v) \cup$ Subtrees hanging from $path(v,x)$\;
		$p\leftarrow path\Big(par\big(par(v)\big),root(\tau)\Big)$
		\tcc*{remaining part of $path\big(r_c,root(\tau)\big)$}
		$p^*\leftarrow path(r_c,x)\cup (x,y)\cup p'_c$\;
		Add $p^*)$ to $T^*$\;
	
		Process-Comp\big($\{p,p_c\setminus p'_c\},\TT\cup\TT_c\setminus\{\tau\},p^* $\big)
		\tcc*{Goes to next stage/phase}
		\caption{DisCon-DFS($r_c,p_c,\TT_c$): Disconnecting Traversal of a component $c$
			having a path $p_c$ and a set of trees $\TT_c$ through the root $r_c$, 
			where either $r_c\in \tau\notin \mathbb{T}_c$ or  $r_c\in T(v_H)$.
		}
		\label{alg:dfs_c2_2}
	\end{procedure*}


\begin{procedure*}
$T(v_H)\leftarrow$ Smallest subtree $\tau'$ of $\tau$, where $|\tau'|>n/2^i$\;

\tcc{Considering Scenario 1.}
$\TT\leftarrow$ Subtrees hanging from $path\big(r_c,r'\big)$ with edge in $p_c$\;
$T(v_L)\leftarrow$ Subtree from $\TT$ containing $v_H$\;
$p^*\leftarrow path\big(r_c,r'\big)$\;
$(x_1,y_1)\leftarrow$ Highest edge to $p^*$ from $\tau'\in\TT$ and $p_c$
\tcc*{where $y_1\in p^*$}
\BlankLine
\If{$x_1\notin T(v_L)$\textbf{or} $x_1\in T(v_H)$ \textbf{or} $x_1=v_L$ \textbf{or} $x_1\in p_c$}
{
Add $p^*$ to $T^*$\;
$\TT\leftarrow$ Subtrees hanging from $p^*$\;
Return Process-Comp\big($\{p_c\},\TT\cup\TT_c\backslash\tau,p^*$\big)\;
\tcc*[f]{Goes to DisConn, DisInt or Path-Halving}	
}

\tcc{Considering Scenario 2.}
$\TT\leftarrow \TT\setminus T(v_L)\cup$ Subtrees hanging from $path(v_L,v_H)$ with edge in $p_c$\;
$(x_d,y_d)\leftarrow$ Highest edge to $p^*$ from $\tau'\in\TT$
\tcc*{where $x_d\in \TT'$}
\lIf{$(x_d,y_d)= \phi$}{$y_d=r_c$}

$(x_p,y_p)\leftarrow$ $\{(x',y'): x'\in T(v_L),y'\in path(y_d,r')$ of minimum $LCA(x',v_H)\}$\;

\BlankLine

$p^*\leftarrow path(r_c,x_p)\cup(x_p,y_p)\cup path\big(y_p,par(v_l)\big)$\;

$\TT\leftarrow$ Subtrees hanging from $path\big(r_c,r'\big)$ with edge in $p_c$\;
$\TT\leftarrow \TT\setminus T(v_L)\cup$ Subtrees hanging from $path(v_L,x_p)$ with edge in $p_c$\;

$(x_2,y_2)\leftarrow$ Lowest edge to $p^*$ from $\tau'\in\TT$ or $p_c$
\tcc*{where $y_2\in p^*$}

$T(v_P)\leftarrow $ The subtree hanging from $path(v_L,x_p)$ having $v_H$\;

\BlankLine
\If{$x_2\notin T(v_P)$\textbf{or} $x_2\in T(v_H)$ \textbf{or} $x_2=v_P$ \textbf{or} $x_2\in p_c$}
{
Add $p^*$ to $T^*$\;
$\TT \leftarrow$ Subtrees hanging from $path(r_c,r')$ and $path(v_L,x_p)$ \;
Return Process-Comp\big($\{p_c, path\big(par(y_p),r'\big)\},\TT\cup\TT_c\backslash\tau,p^*$\big)\;
\tcc*[f]{Goes to DisConn, DisInt or Path-Halving}	
}

\tcc{Considering Scenario 3.}
$\tau_d\leftarrow $ Subtree hanging on $path(v_L,v_H)$ having $x_d$\;
$(x'_2,y'_2)\leftarrow$ Lowest edge from $\tau_d$ to $(r_c,y_p)$\;

\lIf{$y_2$ lower than $y'_2$}{$(x_r,y_r)\leftarrow (x'_2,y'_2)$}
\lElse{$(x_r,y_r)\leftarrow$ $(x_2,y_2)$}

$p^*\leftarrow path(r_c,x_r)\cup (x_r,y_r)\cup path\big(y_r,r'\big)$\;
$\TT\leftarrow$ Subtrees hanging from $path\big(r_c,r'\big)$  with edge to $p_c$\;
$\TT\leftarrow \TT\setminus T(v_L)\cup$ Subtrees hanging from $path(v_L,x_r)$ with edge to $p_c$\;
$(x_3,y_3)\leftarrow$ Lowest edge to $p^*$ from $\tau'\in\TT$ or $p_c$ 
\tcc*{where $y_3\in p^*$}

$T(v_R)\leftarrow $ The subtree hanging from $path(v_L,x_r)$ having $v_H$\;

\BlankLine
\If{$x_3\notin T(v_R)$\textbf{or} $x_3\in T(v_H)$ \textbf{or} $x_3=v_P$ \textbf{or} $x_3\in p_c$}
{
	Add $p^*$ to $T^*$\;
	$\TT \leftarrow$ Subtrees hanging from $path(r_c,r')$ and $path(v_L,x_r)$ \;
	Return Process-Comp\big($\{p_c, path\big(par(v_l),y_r\big)\setminus\{y_r\}\},\TT\cup\TT_c\backslash\tau,p^*$\big)\;
\tcc*[f]{Goes to DisConn, DisInt or Path-Halving}	
}
Heavy-Special($r_c,p_c,\TT_c$);

\caption{Heavy-DFS($r_c,p_c,\TT_c$): Heavy Subtree Traversal of a component $c$
	having a path $p_c$ and a set of trees $\TT_c$ through the root $r_c\in \tau\in \mathbb{T}_c$, where $r'=root(\tau)$ }
\label{alg:dfs_c2}
\end{procedure*}

	\begin{procedure*}
		
		\tcc{Modified $r'$ traversal.}		
		$p^*_{R'}\leftarrow path(r_c,x_2)\cup (x_2,y_2)\cup path\big(y_2,root(\tau)\big)$\;
		$p_1\leftarrow path(par(v_l),y_2)\backslash\{y_2\}$\;
		$\TT\leftarrow$ Subtrees hanging from $path\big(r_c,root(\tau)\big)$  with edge to $p_1$\;
		$\TT\leftarrow \TT\setminus T(v_L)\cup$ Subtrees hanging from $path(v_L,x_2)$ with edge to $p_1$\;
		$(x',y')\leftarrow$ Lowest edge to $p^*_{R'}$ from $\tau'\in\TT$ or $p_1$ 
		\tcc*{where $y'\in p^*_{R'}$}
		
		\BlankLine
		\If{$y'$ at most as high as $y_p$}
		{
			\tcc{Root Traversal of $\tau_d$.}
			$p_{R'}\leftarrow (y_p,x_p)\cup path(x_p,root(\tau_d))$\;
			Add $p^*_{R'}$ and $p_{R'}$ to $T^*$\;
			$\TT \leftarrow$ Subtrees hanging from $path(r_c,root(\tau))\backslash \{T(v_L)\}$ \;
			$\TT \leftarrow$ $\TT\cup$ Subtrees hanging from $path(v_L,x_2)\backslash \{\tau_d\}$ \;
			$\TT \leftarrow$ $\TT\cup$ Subtrees hanging from $path(x_d,root(\tau_d))$ \;
			Return Process-Comp\big($\{p_c, p_1\},\TT\cup\TT_c\backslash\tau,p^*_{R'} \cup p_{R'}$\big)
			\tcc*[f]{Goes to DisConn, DisInt or Path-Halving}	
		}

		\BlankLine
		\tcc{Cover traversal of $p_1$}
		\If(\tcc*[f]{Connected to $p_1$ through $\tau'$}){$x'\notin p_1$}
		{
			$y_{\tau'}\leftarrow x'$\;
			$\{x',x_{\tau'}\}\leftarrow $ Highest edge on $p_1$ from $\tau'$
			\tcc*{where $x'\in p_1$}
		}
		
		\BlankLine
		\uIf{$x'$ at most as high as $y_r$ on $p_1$}
		{
		\tcc{Upward Cover Traversal of $p_1$.}
			\lIf{$x_{\tau'}\neq \phi$}
			{
			$\{x',x_{\tau'}\}\leftarrow $ Lowest edge on $p_1$ from $\tau'$
			\tcc*[f]{where $x'\in p_1$}
			}
		
			$p_{R'}\leftarrow p_{\tau'}\cup path(x',y_2)\backslash \{y_2\}$\;
			$p'_1\leftarrow path(par(v_l),x')\backslash\{x'\}$\;
			$\TT \leftarrow$ Subtrees hanging from $path(v_L,x_2)$ \;
		}
		\Else
		{
		\tcc{Lower Cover Traversal of $p_1$.}
		$(x^*_r,y^*_r)\leftarrow $ 
		Highest edge from $\tau_d$ to $path(par(y_r),y_2)\backslash\{y_2\}$
		\tcc*{where $y^*_r\in path(y_r,y_2)$}		
		\lIf{$(x^*_r,y^*_r)=\phi$}{$(x^*_r,y^*_r)\leftarrow (x_2,y_2)$}
		$p^*_{R'}\leftarrow path(r_c,x^*_r)\cup (x^*_r,y^*_r)\cup path\big(y^*_r,root(\tau)\big)$\;
		$p_{R'}\leftarrow p_{\tau'}\cup path(x',par(v_l))$\;
		$p'_1\leftarrow path(par(x'),y^*_r)\backslash\{y^*_r\}$\;
		$\TT \leftarrow$ $\TT\cup$ Subtrees hanging from $path(v_L,x^*_r)$ \;
		}

		\uIf(\tcc*[f]{Connected to $p_1$ through $\tau'$}){$x_{\tau'}\neq \phi$}
		{
		$p_{\tau'}\leftarrow (y',y_{\tau'})\cup path(y_{\tau'},x_{\tau'})\cup (x_{\tau'},x')$\;
		$p^*_{\tau'}\leftarrow path(LCA(x_{\tau'},y_{\tau'}),root(\tau'))$\;
		$\TT_{\tau'} \leftarrow$ Subtrees hanging from $path(x_{\tau'},y_{\tau'})$ 
						and $path(LCA(x_{\tau'},y_{\tau'}),root(\tau'))$ \;
		}
		\lElse{$p_{\tau'}\leftarrow (y',x')$; $p^*_{\tau'}=\phi$; $\TT_{\tau'}=\phi$}

			Add $p^*_{R'}$ and $p_{R'}$ to $T^*$\;
		$\TT \leftarrow$ $\TT\cup \TT_{\tau'}\cup$
				Subtrees hanging from $path(r_c,root(\tau))\backslash \{T(v_L)\}$ \;
		
		Return Process-Comp\big($\{p_c, p'_1, p^*_{\tau'}\},\TT\cup\TT_c\backslash\tau,p^*_{R'} \cup p_{R'}$\big)\;
			\tcc{Goes to DisConn, DisInt or Path-Halving}	

		\caption{Heavy-Special($r_c,p_c,\TT_c$): Special Case of Heavy Subtree Traversal of a component $c$
			having a path $p_c$ and a set of trees $\TT_c$ through the root $r_c\in \tau\in \mathbb{T}_c$. }
		\label{alg:dfs_c2sp}
	\end{procedure*}

\end{document}